\documentclass[11pt,a4paper]{article}
\usepackage[top=1in, bottom=1in, left=0.85in, right=0.85in]{geometry}
\usepackage{setspace}
\usepackage{setspace}
\usepackage[utf8]{inputenc}
\usepackage[english]{babel}

\usepackage{amsmath,bm,amsthm,amsfonts,authblk}
\usepackage[fleqn]{mathtools}
\usepackage{newcent}
\usepackage[round]{natbib}
\usepackage{graphicx}
\usepackage{epstopdf}
\usepackage{tabularx}
\usepackage{tikz}
\usepackage{blindtext}
\usepackage{colortbl}
\usepackage{adjustbox}
\usepackage{framed}
\usepackage{calc,soul}
\usetikzlibrary{calc}
\usepackage{booktabs}
\usepackage{subcaption}
\usepackage{multirow}
\usepackage{lscape}

\usepackage{rotating}
\usepackage{url,longtable}
\usepackage{footnote}
\usepackage{tablefootnote}
\makesavenoteenv{table}
\usepackage[symbol]{footmisc}

\usepackage{framed}
\definecolor{shadecolor}{RGB}{242,242,255}
\usepackage{datetime}
\usepackage{enumerate}
\usepackage{fdsymbol}

\usepackage{relsize}

\usepackage{appendix}
\usepackage{chngcntr}
\usepackage{etoolbox}
\usepackage{lipsum}
\usepackage{diagbox}
\usepackage{placeins}
\usepackage{float}
\usepackage{nomencl}
\makenomenclature
\usepackage{etoolbox}
\usepackage{hyperref}
\hypersetup{
	colorlinks=true,
	linkcolor=blue,
	filecolor=magenta,      
	urlcolor=cyan,
	citecolor=blue
}

\title{\textbf{Estimation Method under Three-Parameter Generalized Exponential Model: Consistency, Uniqueness and its Applications}}
\author[1]{Kiran Prajapat\thanks{Corresponding author: Kiran Prajapat. Email:kiranprajapat92@gmail.com, kiran.prajapat@newcastle.ac.uk}}
\author[2]{Sharmishtha Mitra}
\author[2]{Debasis Kundu}
\affil[1]{School of Mathematics, Statistics and Physics, Newcastle University, Newcastle, GB}
\affil[2]{Department of Mathematics \& Statistics, Indian Institute of Technology Kanpur, Kanpur, Uttar Pradesh, India - 208016}
\date{}
\begin{document}
	\maketitle 
	\begin{abstract}
		In numerous instances, the generalized exponential distribution can be used as an alternative to the most widely used non-regular family of distributions: Weibull, gamma, lognormal with three-parameters when analyzing lifetime or any skewed continuous data. A non-regular family is a class of probability distributions that do not satisfy the regularity conditions typically assumed in classical statistical inference. Some key features of such family of distributions are: support of its probability density function depends on one its parameters; its likelihood function may not be bounded for a certain range of parameter space, hence maximum likelihood estimators do not exist; the likelihood function even may not be differentiable or integrable as needed, hence Fisher Information 	may not exist or be infinite. Moreover, standard results like MLE existence, consistency, asymptotic normality may fail. Therefore, specialized or robust inferential techniques are needed. This article offers a consistent method for estimating the parameters of a three-parameter generalized exponential distribution that sidesteps the issue of an unbounded likelihood function. The method is hinged on a maximum likelihood estimation of shape and scale parameters that uses a location-invariant statistic. Important estimator properties, such as uniqueness and consistency, are demonstrated for the first time under this approach. In addition, quantile estimates for the assumed distribution are provided. We present a Monte Carlo simulation study along with comparisons to a number of well-known estimation techniques in terms of bias and root mean square error. For illustrative purposes, a real dataset from reliability engineering, has been analyzed and the goodness of fit along with the bootstrap confidence intervals are compared with existing traditional methods. 
	\end{abstract}
	\providecommand{\keywords}[1]{\textbf{\textit{Key words:}} #1}
	\keywords{non-regular family; location-invariant statistics; modified maximum likelihood estimation; consistency; unimodal; confidence interval}
	\doublespacing
	\newtheorem{theorem}{Theorem}[section]
	\newtheorem{corollary}{Corollary}[theorem]
	\newtheorem{proposition}{Proposition}[section]
	\newtheorem{lemma}{Lemma}[section]
	\newtheorem{remark}{Remark}[section]
	\newpage
	\section{Introduction}
	\label{sec1}
	Statistical modeling of any skewed data, such as COVID-19 mortality data, financial weekly return data and electrical lifetime data using a probability model that fits the best is important. The most widely used and well-liked distributions for analyzing such skewed data or lifetime data are the gamma and Weibull distributions with three parameters. These three parameters, which stand for location, scale, and shape, give the distributions a lot of flexibility when it comes to analyzing skewed data. Unfortunately, both distributions have some flaws as well. Additionally, as a special case, the exponentiated Weibull distribution reduces to the three-parameter generalized exponential (GE) distribution when its location parameter is not present. The exponentiated Weibull distribution was first proposed by \cite{mudholkar1995exponentiated}. In many instances, it has been demonstrated that the GE model is applicable as a replacement for the gamma model or the Weibull model; for more information, see \cite{guptakundu1999,gupta2001exponentiated,gupta2007generalized}.
	
	Consider a three-parameter GE distribution denoted by $\text{GE}(\alpha,\beta,\gamma)$, where $\alpha$, $\beta$ and $\gamma$ represent, respectively, the scale, shape, and location parameters. For $\alpha>0,~\beta>0,~\gamma\in {\mathbb{R}}$, a random variable that follows the $\text{GE}(\alpha,\beta,\gamma)$ has the cumulative distribution function (CDF) $F(.;\alpha, \beta, \gamma)$,
	\begin{align}
		\label{eq1}
		F(x;\alpha, \beta, \gamma) = \begin{cases}
			0, & \text{if~} x<\gamma \\
			\big(1-e^{-\frac{x-\gamma}{\alpha}}\big)^\beta, & \text{if~} x>\gamma.
		\end{cases}
	\end{align} 
	Here, all three parameters are undetermined. Since the three-parameter GE distribution belongs to a location-scale family, if $X \sim \text{GE}(\alpha,\beta,\gamma)$, then the standardized random variable $Y=(X-\gamma)/\alpha$ follows a generalized exponential distribution with parameters $(1,\beta,0)$, referred to as the standard GE distribution. For the other similar three-parameter distributions, such as the lognormal, gamma, Weibull, and inverse Gaussian, etc., in which the location parameter is unknown, it is well known that the regularity conditions are not satisfied for the estimation method of the widely recognized maximum likelihood (ML) because the support of the probability density function (PDF) depends on the unknown location parameter, and thus the ML estimation may encounter difficulties. In the majority of time, the maximum likelihood estimator (MLE) does not exist for a particular range of the parameter space. In such situations, the likelihood becomes unbounded. Some additional issues with MLEs for non-regular distributions (when they exist): Asymptotic normality of the MLE and its functions cannot be used; the Fisher information matrix cannot be used for asymptotic variances and covariances due to the failure of the regularity conditions. Several authors, including \cite{cohen1980estimation,cohen1982modified,cohen1984modified,smith1985maximum,smith1987comparison,hirose1997inference,hall2005bayesian,nagatsuka2012parameter,nagatsuka2013consistenta,nagatsuka2014consistent,nagatsuka2015efficient,prajapat2021consistent,basu2023three} have explored this issue. The identical issue arises in the case of three-parameter GE distribution. \cite{guptakundu1999} elaborated on the ML estimation for the three-parameter GE distribution. When the shape parameter  $\beta < 1$, it has been demonstrated that the MLE does not exist because the likelihood function becomes unbounded when the location parameter $\gamma$ is approximately equal to the smallest observation in the observed sample; whereas asymptotic results have been presented only for the range of $\beta > 2$. Taking this problem into account, \cite{prajapat2021consistent} and \cite{basu2023three} offered other procedures for estimating the parameters of three-parameter GE for its entire parameter space.
	
	Most recent and newly proposed estimation methods in this direction include the location-scale-parameter-free (LSPF) and location-parameter-free (LPF) methods. 
	As the names of the methods imply, the LSPF method is based on a location and scale invariant statistic, whereas the LPF method uses a location invariant statistic; consequently, the likelihood functions based on the invariant statistics for these methods are one-dimensional and two-dimensional, respectively. 
	Uniqueness and the consistency properties of LPF estimators have not yet been demonstrated analytically for any of the well-known three-parameter distributions examined in the literature. We assume this since the resultant likelihood function has a complicated form. Consequently, an attempt has been made for the hypothesized GE distribution, and consistency is proved whereas unimodality has been established for some cases. 
	 
	This study shows that the LPF approach has an advantage over the LSPF method in that it requires less time to execute the simulations and has less computational complexity.  The main explanation of this is the reduction in the number of integration as compare to the LSPF method. One of the LSPF method's drawbacks over the LPF method is that it uses the estimates from the preceding sequence to construct estimates separately and sequentially. This might lead to a significant build-up of biases in the sequences. The LPF, on the other hand, produces estimates through two-dimensional optimization by reducing number of steps in the estimation procedure. Moreover, LPF method works with a reduced number of degenerate random variables as compared to the LSPF method, which possibly may encourage it to better perform in some situations.
	
	Therefore, the major purpose of the study is to develop the LPF method of estimation in detail for the parameters and quantiles of the three-parameter GE distribution. The second key purpose of the paper is to demonstrate the estimators' uniqueness and consistency. Because proving the uniqueness analytically is challenging, therefore proofs are provided for some particular cases. To obtain bias and root mean square error (RMSE) of the estimators, a Monte Carlo simulation is implemented. On the basis of the biases and RMSE of the estimators, we also conduct a comparative study with some prominent methods.
	
	The remaining sections are structured as follows. The estimation procedure based on the proposed LPF method for the three-parameter GE distribution is detailed in Section \ref{sec2}. In this section, we also discuss the estimators' properties, such as their uniqueness and consistency. In Section \ref{section3}, a Monte Carlo simulation study is conducted for the purpose of evaluating the LPF method and making comparisons to some existing methods. In Section \ref{section4}, the LPF method along with other methods are illustrated using a real-world dataset of electrical lifetimes. Section \ref{section5} presents concluding remarks while summarizing the study.
	
	\section{Proposed Estimators and their Properties} 
	\label{sec2}
	Assume that $X_1, X_2, \dots, X_n$ are $n$ independent and identically distributed (\textit{i.i.d.}) random variables following the three-parameter GE distribution with the common CDF defined in equation $\eqref{eq1}$. Throughout the paper, it is presumed that $n \geq 3$. Consider the order statistics of $X_1, X_2, \dots, X_n$ to be $X_{(1)} < X_{(2)} < \dots < X_{(n)}$. 
	\noindent To begin the estimation procedure, we will consider the following statistic:
	\begin{align}
		\label{eq3}
		V_{(i)} = X_{(i)} - X_{(1)}, ~ i=1,2,\dots,n.
	\end{align} 
	$V_{(i)}$'s probability distribution is independent of the location parameter. It is worth noting that $V_{(1)} = 0$. The scale and shape parameters are therefore estimated based on the transformed data $V_{(1)}, V_{(2)}, \dots, V_{(n)}$, whose joint probability distribution primarily depends on scale and shape parameters.
	
	In this estimating sequence according to the developed LPF approach,  estimators  of the scale and shape parameter are then being used to estimate the location parameter. The estimation of parameters is discussed in detail in Subsection \ref{subsec2.1}, while the estimation of quantiles for the lifetime distribution is elaborated on in Subsection \ref{subsec2.3}. 
	\subsection{Parameter Estimation }
	\label{subsec2.1}
	We estimate the scale and shape parameters $\alpha$ and $\beta$ based on the random variables $V_{(1)}, V_{(2)}, \dots, V_{(n)}$. As the likelihood function of $\alpha$ and $\beta$ based on the transformed data is not dependent on the location parameter, it is bounded. Let $v_1, v_2, \dots, v_n$ represent the respective realizations of $V_{(1)}, V_{(2)}, \dots, V_{(n)}$. Note that $v_1$ must be $0$.
	\begin{theorem}
		\label{thm_2.1}
		The likelihood function of $\alpha$ and $\beta$, given $ v_2, v_3, \dots, v_{n} $, is given by
		\begin{align}
			\label{eq4}
			\ell_v(\alpha, \beta|v_2, \dots, v_{n}) = n! \Big(\frac{\beta}{\alpha}\Big)^n \int_{0}^{\infty}  e^{ -\frac{1}{\alpha}\sum_{i=1}^{n} (u+v_i) } \prod_{i = 1}^{n} \big( 1 - e^{-\frac{u+v_i}{\alpha}} \big)^{\beta-1}  du, ~ \alpha > 0,~ \beta > 0, 
		\end{align}
		with $0 < v_{2} < \dots < v_{n} < \infty,~ v_{1} = 0$.
		\begin{proof}
			See Appendix \ref{app_1}.
		\end{proof}	
	\end{theorem}
	In Theorem \ref{thm_2.1}, the likelihood function is a bounded and differentiable function with respect to the parameters $\alpha$ and $\beta$. These properties' proofs are listed in Appendix \ref{app_2}.  In order to maintain simplicity, we will refer to $\ell_v(\alpha, \beta|v_2, \dots, v_{n})$ as $ \ell_v(\alpha, \beta)$  from now on.
	
	Due to the complexity of the likelihood function, we were unable to prove the unimodality of the bivariate function $\ell_v(\alpha, \beta)$, but we could establish the unimodality of the likelihood function $\ell_v(\alpha, \beta)$ when one of the parameters $\alpha$ and $\beta$ is fixed. Now, we present two main findings in the subsequent theorems, the first of which relates to unimodality and the second to the consistency of the unique maximum. In Theorem \ref{thm_2.2}, the unimodality of the likelihood function is emphasized.
	\begin{theorem}
		\label{thm_2.2}
		For every $0 < v_{2} < \dots < v_{n} <  \infty, v_{1} = 0$, the likelihood function $\ell_v(\alpha, \beta)$ is unimodal function of $\alpha>0~(\text{or} ~\beta>0)$ whenever $\beta~( \text{or} ~ \alpha)$ is fixed.
		\begin{proof}
			Let us recall the likelihood function in equation \eqref{eqq3} from Appendix \ref{app_2} and rewrite the function $h_v(\alpha, \beta;u)$ defined in there:
			\begin{align*}
				\ell_v(\alpha, \beta) = &  n!  \int_{0}^{\infty}  e^{ h_v(\alpha, \beta;u)}  \ du, ~ \alpha > 0,~ \beta > 0, & \nonumber \\
				\text{where} ~ h_v(\alpha, \beta;u) = &   n \ln(\beta) - n \ln(\alpha) - \frac{1}{\alpha} \sum_{i=1}^{n} c_i + (\beta - 1) \sum_{i=1}^{n} \ln(1 - e^{-c_i/\alpha}) & \nonumber
			\end{align*} with $c_i = u+v_i$. 
			Since the likelihood is differentiable with respect to $\alpha$ and $\beta$, therefore let us recall its derivatives, from equations \eqref{eqq5} and \eqref{eqq10}, which are as follows.
			\begin{align}
				\label{eq8}
				\frac{\partial \ell_v(\alpha, \beta)}{\partial \alpha}  =  & n! \int_0^\infty \frac{\partial h_v(\alpha, \beta;u)}{\partial \alpha}  e^{ h_v(\alpha, \beta;u)}  \ du, & 
			\end{align} with $\frac{\partial h_v(\alpha, \beta;u)}{\partial \alpha} = -\frac{n}{\alpha} + \frac{1}{\alpha^2} \sum\limits_{i=1}^{n} c_i \Big( \frac{1 - \beta e^{-c_i/\alpha}}{1 - e^{-c_i/\alpha}} \Big) = -\frac{n}{\alpha} + \frac{1}{\alpha^2} \sum\limits_{i=1}^{n} c_i \Big( 1 +  \frac{1 - \beta }{e^{c_i/\alpha} - 1} \Big), $ and
			\begin{align}
				\label{eq9}
				\frac{\partial \ell_v(\alpha, \beta)}{\partial \beta}  = & n! \int_0^\infty \frac{\partial h_v(\alpha, \beta;u)}{\partial \beta}  e^{ h_v(\alpha, \beta;u)}  \ du, &
			\end{align} with $\frac{\partial h_v(\alpha, \beta;u)}{\partial \beta} = \frac{n}{\beta} + \sum\limits_{i=1}^{n} \ln( 1 - e^{-c_i/\alpha} ). $ 
			
			First we show that $\ell_v(\alpha, \beta)$ is a unimodal function of $\beta$ for a fixed $\alpha$, and then its unimodality with respect to $\alpha$ when $\beta$ is fixed.  Therefore, let us first assume that $\alpha$ is fixed. Note that, for any choice of $u$ and $v_i$'s, $ e^{ h_v(\alpha, \beta;u)} > 0 $ and also the integrals are on positive support, hence the change in the sign of $ \frac{\partial \ell_v(\alpha, \beta)}{\partial \beta}  $ in \eqref{eq9} directly depends on the change in the sign of $ \frac{\partial h_v(\alpha, \beta;u)}{\partial \beta} $.
			$ \frac{\partial h_v(\alpha,\beta;u)}{\partial \beta}   \to \infty$ as $\beta \downarrow 0$ and $\frac{\partial h_v(\alpha,\beta;u)}{\partial \beta}  < 0 $ as $\beta \to \infty$. Moreover, $\frac{\partial^2 h_v(\alpha,\beta;u)}{\partial \beta^2}  = - \frac{n}{\beta^2} < 0 ~\forall ~ \beta $, \textit{i.e.}, $\frac{\partial h_v(\alpha,\beta;u)}{\partial \beta} $ changes sign from positive to negative and the change in sign is only once. Hence, for a fixed $\alpha$, $\frac{\partial \ell_v(\alpha,\beta)}{\partial \beta}  $ changes sign in the similar way and $\frac{\partial \ell_v(\alpha,\beta)}{\partial \beta} = 0$ has a unique solution which maximizes the likelihood function $\ell_v(\alpha,\beta)$ with respect to $\beta$. 
			Now, we move forward to show that $\ell_v(\alpha, \beta)$ is a unimodal function of $\alpha$ for a fixed $\beta$. With the same argument as given above, sign of $ \frac{\partial \ell_v(\alpha, \beta)}{\partial \alpha}  $ in \eqref{eq8} directly depends on the sign of $ \frac{\partial \ell_v(\alpha, \beta)}{\partial \alpha} $ or else, we can say that sign of $ \frac{\partial \ell_v(\alpha, \beta)}{\partial \alpha}$ is inversely proportional to the sign of $\alpha - \frac{1}{n} \sum\limits_{i = 1}^{n} c_i \Big(1 + \frac{   1 - \beta }{ e^{ c_i/\alpha } - 1} \Big)$. Say 
			\begin{align}
				\label{eq10}
				H_{1,n}(\alpha) = \frac{1}{n} \sum\limits_{i = 1}^{n} c_i \Big(1 + \frac{   1 -\beta}{ e^{ c_i/\alpha } - 1} \Big).
			\end{align} 
			When $\beta > 1 $, $H_{1,n}(\alpha)$ is a strictly decreasing function in $\alpha$ that decreases from $\frac{1}{n} \sum\limits_{i = 1}^{n}c_i$ to $-\infty$. $H_{1,n}(\alpha)$ is a constant function of $\alpha$ whenever $\beta = 1$ taking value  $\frac{1}{n} \sum\limits_{i = 1}^{n}c_i$. It is easy to see that $H_{1,n}(\alpha)$ and $\alpha$ meet exactly once whenever $\beta \geq 1$. Therefore, $\alpha - H_{1,n}(\alpha)$ changes sign from negative to positive and change in sign is only once whenever $\beta \geq 1$. Again, when $\beta < 1 $, $H_{1,n}(\alpha)$ is strictly increasing in $\alpha$ and it increases from $\frac{1}{n} \sum\limits_{i = 1}^{n} c_i$ to $\infty$. Now, let us consider the quantity
			\begin{align*}
				\frac{\partial H_{1,n}(\alpha)}{\partial \alpha} = \frac{ 1 - \beta}{n\alpha^2} \sum\limits_{i = 1}^{n} \Bigg(\frac{ c_i^2 e^{ -c_i/\alpha } }{ (1 - e^{ -c_i/\alpha })^2 } \Bigg).
			\end{align*} 
			Since $\frac{ s^2 e^{ -s } }{ (1 - e^{ -s })^2 } < 1 ~ \forall ~ s > 0 $ (see part $(\textit{i})$ of Lemma $2$ in \cite{ghitany2013}), we have $\frac{\partial H_{1,n}(\alpha)}{\partial \alpha} < 1$. 
			Hence,  $\alpha$ and $H_{1,n}$ have to meet exactly once. This implies that $\alpha - H_{1,n}(\alpha)$ changes its sign only once from negative to positive for $\beta < 1$ also. Since the sign of $ \frac{\partial \ell_v(\alpha, \beta)}{\partial \alpha}$ is inversely proportional to the sign of $\alpha -  H_{1,n}(\alpha)$, it is clear that change in sign of $ \frac{\partial \ell_v(\alpha, \beta)}{\partial \alpha}$ from positive to negative and it is only once. It implies that $\ell_v(\alpha, \beta)$ is unimodal with respect to $\alpha$ for a fixed $\beta$. 
		\end{proof}
	\end{theorem}
	Now, in Theorem \ref{thm2.3}, we demonstrate the consistency of the estimators of $\alpha$ and $\beta$ obtained by maximizing $\ell_v(\alpha, \beta)$. 
	\begin{theorem}
		\label{thm2.3}
		Estimators based on the maximization of the likelihood $\ell_v(\alpha, \beta)$ are consistent estimators for $\alpha>0$ and $\beta > 0$.
		\begin{proof}
			To show the consistency of the estimators of $(\alpha, \beta)$ obtained by maximizing the likelihood function $\ell_v(\alpha,\beta)$, it suffices to prove the following result: \\
			For any fixed $\alpha \neq \alpha_0$ and $\beta \neq \beta_0$, where $\alpha_0$ and $\beta_0$ are true values of the parameters $\alpha$ and $\beta$,
			\begin{equation*}
				\lim_{n \to \infty} P \Big( \frac{\ell_v(\alpha,\beta;V_{(2)}, \dots, V_{(n)})}{\ell_v(\alpha_0,\beta_0;V_{(2)}, \dots, V_{(n)})} < 1 \Big) = 1 .
			\end{equation*}
			To proceed further, we recall the conditional joint PDF of the order statistics of $V_{(2)}, \dots, V_{(n)}$ given $Z_{(1)} = u$ (See Appendix \ref{app_1}). Here $Z_{(1)} = {X_{(1)} - \gamma_0}$ and $\gamma_0$ is true value of the parameter $\gamma$. The conditional joint PDF is as follows:
			\begin{align*}
				f_{V_{(2)}, \dots, V_{(n)}| Z_{(1)} = u}(v_2, \dots, v_{n} | u) & =  \frac{f_{Z_{(1)},V_{(2)}, \dots, V_{(n)}}(u, v_2, \dots, v_{n})}{f_{Z_{(1)}}(u)} & \\
				& = \frac{ n! g(u;\alpha_0, \beta_0) \Big\{ \prod_{i = 2}^{n} g(u+v_i;\alpha_0, \beta_0) \Big\} }{ n  g(u;\alpha_0, \beta_0) (1 - G(u;\alpha_0, \beta_0))^{n-1} } & \\
				& = (n-1)! \prod_{i=2}^{n} \frac{ g( u+v_i;\alpha_0, \beta_0)}{1 - G(u;\alpha_0, \beta_0)} & \\
				& = (n-1)! \prod_{i = 2}^{n} f_{V_i| Z_{(1)} = u} (v_{i}|u),&  
			\end{align*}
			where $f_{V_i| Z_{(1)} = u} (v_{i}|u)  = \frac{  g( u+v_i;\alpha_0, \beta_0)}{1 - G(u;\alpha_0, \beta_0)}$. Here, $V_{i}'s$ are \textit{i.i.d.} conditional on $ Z_{(1)} = u$, with a common conditional PDF $f_{V_i| Z_{(1)} = u} (.|u)$. Define 
			\begin{align*}
				\ell_{u}(\alpha,\beta;V_{(2)}, \dots, V_{(n)}) & = (n-1)! \prod_{i=2}^{n} \frac{  g( u+V_i;\alpha, \beta)}{1 - G(u;\alpha, \beta)} &
			\end{align*}
			and then, consider the following quantity for fixed $u>0$, conditioned on $Z_{(1)} = u$. For every $\alpha \neq \alpha_0$ and $\beta \neq \beta_0$,
			\begin{align}
				\label{eqq1}
				\frac{1}{n-1} & \ln\bigg\{ \frac{\ell_{u}(\alpha,\beta;V_{(2)}, \dots, V_{(n)})}{\ell_{u}(\alpha_0,\beta_0;V_{(2)}, \dots, V_{(n)})}\bigg\}  = \frac{1}{n-1} \sum_{i = 2}^{n} \ln \Bigg\{ \frac{  g( u+V_i;\alpha, \beta) / (1 - G(u;\alpha, \beta))}{ g( u+V_i;\alpha_0, \beta_0) / (1 - G(u;\alpha_0, \beta_0))} \Bigg\}. & 
			\end{align}
			By the law of large numbers, right hand side of the equation \ref{eqq1} converges to 
			\begin{align*}
				E  \Bigg(  \ln \Bigg\{ \frac{  g( u+V;\alpha, \beta) / (1 - G(u;\alpha, \beta))}{ g( u+V;\alpha_0, \beta_0) / (1 - G(u;\alpha_0, \beta_0))} \Bigg\} \bigg|  Z_{(1)} = u \Bigg). 
			\end{align*}
			Here, $V$ has a conditional PDF $f_{V| Z_{(1)} = u} (.|u)$ given $Z_{(1)} = u$. Now, by Jensen's inequality, 
			\begin{align*}
				E  \Bigg(  & \ln \Bigg\{ \frac{  g( u+V;\alpha, \beta) / (1 - G(u;\alpha, \beta))}{ g( u+V;\alpha_0, \beta_0) / (1 - G(u;\alpha_0, \beta_0))} \Bigg\} \bigg|  Z_{(1)} = u \Bigg) & \\
				& \leq \ln  \Bigg\{  E \Bigg( \frac{ g( u+V;\alpha, \beta) / (1 - G(u;\alpha, \beta))}{ g( u+V;\alpha_0, \beta_0) / (1 - G(u;\alpha_0, \beta_0))} \bigg|  Z_{(1)} = u \Bigg) \Bigg\} & \\
				& = \ln \Bigg\{  \int_{0}^{\infty} \frac{ g( u+v;\alpha, \beta)  }{(1 - G(u;\alpha, \beta))} \ d v \Bigg\} = 0. & 
			\end{align*}
			Therefore, it implies that
			\begin{align}
				\label{eqq2}
				& \lim_{n \to \infty} P  \Bigg(\ \frac{1}{n-1}\ln\bigg\{ \frac{\ell_{u}(\alpha,\beta;V_{(2)}, \dots, V_{(n)})}{\ell_{u}(\alpha_0,\beta_0;V_{(2)}, \dots, V_{(n)})}\bigg\} < 0 \ \bigg| \ Z_{(1)} = u\Bigg) = 1 & \nonumber 
			\end{align}
		    \begin{align}
				& \implies \lim_{n \to \infty} P  \Bigg(\  \frac{\ell_{u}(\alpha,\beta;V_{(2)}, \dots, V_{(n)})}{\ell_{u}(\alpha_0,\beta_0;V_{(2)}, \dots, V_{(n)})}  < 1 \ \bigg|  \ Z_{(1)} = u\Bigg) = 1 & \nonumber \\
				& \implies 	\lim_{n \to \infty} P \big( \ell(\alpha,\beta;V_{(2)}, \dots, V_{(n)}) <  \ell(\alpha_0,\beta_0;V_{(2)}, \dots, V_{(n)}) |  Z_{(1)} = u \big) = 1. &
			\end{align}	
			Moreover
			\begin{align}	
				& P (\ell(\alpha,\beta;V_{(2)}, \dots, V_{(n)}) <  \ell(\alpha_0,\beta_0;V_{(2)}, \dots, V_{(n)})) & \nonumber \\
				& = \int_0^\infty  P \big( \ell(\alpha,\beta;V_{(2)}, \dots, V_{(n)}) <  \ell(\alpha_0,\beta_0;V_{(2)}, \dots, V_{(n)}) |  Z_{(1)} = u \big)  f_{ Z_{(1)}  } (u)  \ du & \nonumber 
			\end{align} Therefore, using the Lebesgue's dominated convergence theorem and \eqref{eqq2}
			\begin{align}
				& \lim_{n \to \infty} P\big( \ell(\alpha,\beta;V_{(2)}, \dots, V_{(n)}) <  \ell(\alpha_0,\beta_0;V_{(2)}, \dots, V_{(n)}) \big) & \nonumber \\
				& =  \int_0^\infty \lim_{n \to \infty} P \big( \ell(\alpha,\beta;V_{(2)}, \dots, V_{(n)}) <  \ell(\alpha_0,\beta_0;V_{(2)}, \dots, V_{(n)}) |  Z_{(1)} = u \big)  f_{ Z_{(1)}  } (u)  \ du & \nonumber \\
				& = \int_0^\infty \lim_{n \to \infty}  f_{ Z_{(1)}  } (u)  \ du = \lim_{n \to \infty} \int_0^\infty  f_{ Z_{(1)}  } (u)  \ du  = 1.	&
			\end{align} Hence the result is proved.
		\end{proof}
	\end{theorem}
	We already have estimators based on the maximization of the likelihood function $\ell_v(\alpha, \beta)$ using the location-invariant data $v_2, v_3, \dots, v_{n}$. Let us refer to the estimators of $\alpha$ and $\beta$ based on the location-invariant data as $\widehat{\alpha}_v$ and $\widehat{\beta}_v$, respectively.
	
	In addition, it is appropriate to use $X_{(1)}$ as an estimator for the location parameter $\gamma$. Let us represent it as $\widehat{\gamma}_{init}$ \textit{i.e.}, $ \widehat{\gamma}_{init} = X_{(1)} $. Because $E(X_{(1)}) = \gamma + \alpha \int_{0}^{\infty} (1 - F_Y(y; \beta))^n\ dy$ with $Y \sim \text{GE}(1,\beta,0)$, it should be emphasized that $\widehat{\gamma}_{init}$ is a biased estimator of the parameter $\gamma$. It is therefore feasible to look at a bias-corrected estimate of $\gamma$, denoted by $\widehat{\gamma}_{v}$, as $X_{(1)} -  \widehat{\alpha}_v \int_{0}^{\infty} (1 - F_Y(y;\widehat{\beta}_v))^n\ dy$. 
	
	Now, we establish that $\widehat{\gamma}_v$ is a consistent estimator of $\gamma$. To explain the consistency of $\widehat{\gamma}_{init}$, consider the following probability for an arbitrary $\epsilon>0$:
	\begin{align}
		\label{eq_4}
		\text{P}(|\widehat{\gamma}_{init} - \gamma|> \epsilon) = \text{P}(X_{(1)} - \gamma > \epsilon) = \text{P}(n^{1/\beta}\frac{X_{(1)} - \gamma}{\alpha} > n^{1/\beta} \frac{\epsilon}{\alpha})
	\end{align} Using the Theorem $8$ of \cite{guptakundu1999}, it can be observed that the probability aforesaid in \eqref{eq_4} is less than $e^{-n (\epsilon / \alpha )^\beta }$ and  converges to $0$ as sample size $n$ approaches $\infty$, proving that $\widehat{\gamma}_{init}$ is consistent for $\gamma$. Recall, in order to show the consistency of the bias-corrected estimator, that $\widehat{\gamma}_v = X_{(1)} - \widehat{\alpha}_v \int_{0}^{\infty} (1 - (1 - e^{-y})^{\widehat{\beta}_v})^n\ dy$.  Using Slutsky's theorem and the facts that $\widehat{\alpha}_v$ and $\widehat{\beta}_{v}$ are consistent for $\alpha$ and $\beta$, respectively, it can be shown that $\widehat{\alpha}_v \int_{0}^{\infty} (1 - (1 - e^{-y})^{\widehat{\beta}_v})^n\ dy $ converges to $0$ in probability. Consequently, applying Slutsky's theorem once more,  $\widehat{\gamma}_v$ is consistent for $\gamma$.
	
	\noindent The above-mentioned detailed estimation approach under the LPF method described in this section is quickly summarized as follows:
	\begin{description} \vspace{-1em}
		\setlength{\itemsep}{0em}
		\item[Step 1.] Obtaining $\widehat{\alpha}_v$ and $\widehat{\beta}_v$, take $\widehat{\gamma}_{init} = X_{(1)}$.  
		\item[Step 2.] Utilize $\widehat{\alpha}_v$, $\widehat{\beta}_v$ and $\widehat{\gamma}_{init}$ of Step 1 to obtain the bias-corrected $\widehat{\gamma}_{v}$ as follows: \vspace{-1em}
		\begin{align*} 
			\widehat{\gamma}_v = & X_{(1)} -  \widehat{\alpha}_v \int_{0}^{\infty} (1 - (1 - e^{-y})^{\widehat{\beta}_v})^n\ dy.&
		\end{align*} 
	\end{description}
	\subsection{Quantile Estimation}
	\label{subsec2.3}
	We consider the estimation of the $\zeta$-th quantile of the three-parameter GE distribution with CDF given in \eqref{eq1}. 
	For $0<\zeta<1$, the $\zeta$-th quantile is the solution to the equation 
	$F(x_\zeta;\alpha,\beta,\gamma)=\zeta$ with respect to $x_\zeta$, which yields
	\begin{align}
		\label{eq11}
		x_\zeta = \gamma - \alpha \ln\left(1-\zeta^{1/\beta}\right).	
	\end{align} 
	In this study, the quantile estimator is obtained by substituting the LPF estimators of $\alpha$, $\beta$, and $\gamma$ into \eqref{eq11}. The following theorem establishes the consistency of the plug-in estimator $\widehat x_{\zeta,v}$.
	\begin{theorem}
		\label{thm_2.4}
		For a fixed $\zeta\in(0,1)$, and parameters $\alpha>0$, $\beta>0$, and $\gamma\in\mathbb{R}$, the $\zeta$-th quantile of the GE distribution is
		$ x_\zeta(\alpha,\beta,\gamma) = \gamma - \alpha \ln\!\left( 1 - \zeta^{1/\beta} \right)$ and its plug-in estimator based on the LPF estimators is
		\begin{align*}
			\widehat x_{\zeta,v}
			= x_\zeta(\widehat{\alpha}_v,\widehat{\beta}_v,\widehat\gamma_v)
			= \widehat\gamma_v
			- \widehat\alpha_v \ln\!\left( 1 - \zeta^{1/\widehat\beta_v} \right).
		\end{align*}
		Then
		\begin{align*}
			\widehat x_{\zeta,v} \xrightarrow{P} x_\zeta(\alpha,\beta,\gamma),
		\end{align*} where $\xrightarrow{P}$ denotes convergence in probability.
	\end{theorem}
	\begin{proof}
		Define the mapping $g:(0,\infty)\times(0,\infty)\times\mathbb{R}\to\mathbb{R}$ by $g(\alpha,\beta,\gamma)
		= \gamma - \alpha \ln\!\left( 1 - \zeta^{1/\beta} \right)$. For any $\beta>0$ and $\zeta\in(0,1)$ we have $0<\zeta^{1/\beta}<1$, and therefore
		$1 - \zeta^{1/\beta} \in (0,1)$ and its logarithm is finite. 
		Thus $g$ is well defined and continuous for all $(\alpha,\beta,\gamma)$ with $\beta>0$.\\
		\noindent From the consistency results of the LPF estimators established in 
		Section~\ref{subsec2.1}, we have
		\begin{align*}
			(\widehat\alpha_v,\widehat\beta_v,\widehat\gamma_v)
			\xrightarrow{P} (\alpha,\beta,\gamma),
		\end{align*}
		and since $g$ is continuous at $(\alpha,\beta,\gamma)$, the Continuous Mapping Theorem (Theorem 1.10 in \cite{shao2003mathematical}) implies
		\begin{align*}
			\widehat x_{\zeta,v}
			= g(\widehat\alpha_v,\widehat\beta_v,\widehat\gamma_v)
			\xrightarrow{P}
			g(\alpha,\beta,\gamma)
			= x_\zeta(\alpha,\beta,\gamma).
		\end{align*}
		This proves the result.
	\end{proof}
	\begin{remark}
		The above theorem requires $\zeta\in(0,1)$. For $\zeta=0$, the quantile equals $\gamma$, so $\widehat x_{0,v} = \widehat\gamma_v \xrightarrow{P} \gamma$. For $\zeta$ very close to $1$, numerical instability may arise, although the plug-in estimator remains consistent.
	\end{remark}
	\section{Simulation Study}
	\label{section3}
	We evaluate the performance of the proposed estimator through a Monte Carlo simulation study. The proposed method of estimation is the LPF method. We compare the performance of the LPF method to the LSPF method and two other modified maximum likelihood estimation methods, because these methods provide estimators for the entire parameter space. Since the MLE does not exist when $\beta< 1$, the LPF method is compared to the MLE when the shape parameter is assumed to be greater than or equal to $1$. The comparisons are based on the estimators' biases and the RMSE. These two modified maximum likelihood estimators (MMLEs) are MMLE I and MMLE III, that are discussed in Section 3 of  \cite{prajapat2021consistent}. Recall that MMLE I is the most convenient estimation method, which has been recently implemented by \cite{pasari2014three} and \cite{raqab2008estimation}, whereas the MMLE III method was proposed by \cite{hall2005bayesian}. In this article, the MMLE I method is modified using the bias-corrected estimator of $\gamma$ before being implemented.
	
	All results are obtained using the Monte Carlo simulation with the shape parameter $\beta$ set to $ 0.50,~ 0.75,~ 1.00, ~1.50, ~2.00,~ 3.00$ and the sample size $n$ set to $20,~ 50$, $100$ and $200$. All simulation results  in terms of biases and RMSEs are provided based on $10,000$ simulations. Since the GE distribution belongs to a location-scale family, simulations are conducted under the standardized parameterization $\alpha=1$ and $\gamma=0$. Results for general $\alpha>0$ and $\gamma\in\mathbb{R}$ follow by simple location-scale transformation. We present bias and RMSE for the estimators of location, scale, and shape parameters in Subsection \ref{subsec3.1} and for the estimators of quantiles in Subsection \ref{subsec3.2}.    
	
	During the simulation study, a certain proportion of the generated datasets has to be rejected. This proportion is reported in the last column of the corresponding tables. These rejections occur when the resulting parameter estimates are clearly unreasonable or numerically unstable, for example, an estimate such as $10^7$ for the shape parameter. Such values indicate that the corresponding dataset is not suitable for evaluating estimators' performance and therefore must be excluded. To implement this, we choose an upper cutoff value for the shape parameter estimate, denoted by $\beta_{U}$. This threshold is determined, for instance, by examining the histogram of the simulated estimates and identifying a reasonable upper bound beyond which the estimates are implausible. Therefore, the proportion of rejection is defined as
	\begin{align*}
		p = \frac{\text{Number of samples with } \widehat{\beta}_{v} \ge \beta_{U}}
		{\text{Total number of simulations}}.
	\end{align*}
	Since estimation in three-parameter models for very small samples (e.g., $n \le 10$) is well known to produce highly unstable estimates with large biases, very high variability, and frequent non-convergence (see \cite{gupta2001generalized}), we focus on moderate and larger sample sizes $n = 20, 50, 100, 200$, as we observe similar issues for such small samples. Moreover, for such very small sample sizes, the rejection proportion $p$ also becomes extremely high. Furthermore, in the case of the LSPF method, numerical results for $n=200$ are omitted due to a computational issue, where the likelihood calculations exceed machine precision and produce values that are out of bounds during the simulations. 
	
	We also evaluate the computational efficiency of the LPF method, since computation time is an important practical consideration. The total time required for any estimation method depends on several factors, such as the value of the shape parameter~$\beta$, the sample size~$n$, and the number of Monte Carlo replications. Among these, the sample size has the strongest influence, and therefore we report computation times for different values of $n$.
	
	All computations were carried out in \texttt{R}, using the \texttt{foreach} and \texttt{doParallel} packages for parallel execution on a high-performance computing (HPC) cluster with 40 processing cores. For each method, we generate $25,000$ samples, and the first $10,000$ samples are used to obtain the estimates. Thus, the computational times reported in Table~\ref{tab:comp_time} are compared based on $25,000$ Monte Carlo iterations. The \texttt{R} implementations of the MLE, MMLE I, and MMLE III methods produce numerical results within seconds or minutes, and therefore their computational cost is negligible compared to that of the LSPF and LPF methods.
	\begin{table}[ht]
		\centering
		\caption{Computational time (LPF vs.\ LSPF) for different sample sizes $n$ at $\beta = 1.50$.}
		\label{tab:comp_time}
		\begin{tabular}{lllll}
			\toprule
			Method & $n=20$ & $n=50$ & $n=100$ & $n=200$ \\ 
			\midrule
			LSPF  & 18.50 hrs & 27.07 hrs & 33.90 hrs & --   \\
			LPF   &  2.80 mins & 3.20 mins & 3.60 mins & 6.00 mins \\
			\bottomrule
		\end{tabular}
	\end{table}
	The results in Table~\ref{tab:comp_time} clearly show the substantial computational advantage of the LPF method. For instance, when $\beta = 1.50$ and $n = 100$, the LSPF method requires approximately 33.90 hours, whereas the LPF method completes the same number of replications in only 3.60 minutes. This significant reduction in computational time demonstrates the practical benefit of the LPF method, particularly for large sample sizes or simulation-intensive studies. All \texttt{R} source codes used to generate the simulation results and data analyses are provided as supplementary material to ensure full reproducibility.
	\subsection{Evaluation of Parameter Estimation} 
	\label{subsec3.1}
	In Tables \ref{table1}-\ref{table2}, all numerical simulations for estimating the parameters are presented.  The shape parameter is assumed to be greater than $1$ in Table \ref{table2} so that the MLEs are computed sensibly. To get a clear picture of the performance of the estimates of all three parameters, their biases and RMSEs are plotted in Figures \ref{fig_1}-\ref{fig_3} against varying $\beta$ values. Now, we will attempt to summarize the results of the simulation study performed. The following observations are based on the simulation study reported in the first two tables:
	\begin{figure}
		\centering
		\begin{shaded}
			\begin{adjustbox}{width=1\textwidth,height=0.55\textwidth}
				\includegraphics[width=\textwidth]{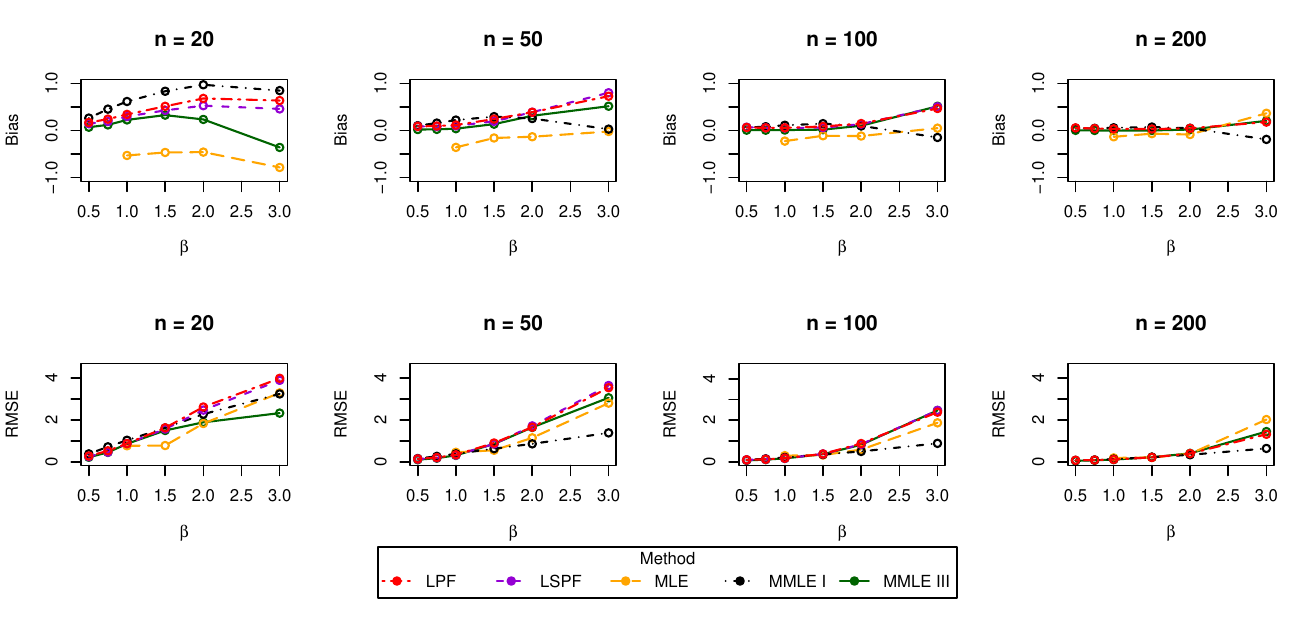}
			\end{adjustbox}
		\end{shaded}
		\caption{Plots for bias and RMSE of the estimator of shape parameter based on various estimation methods varying $\beta$ values.}
		\label{fig_1}
	\end{figure}
	\begin{figure}
		\centering
		\begin{shaded}
			\begin{adjustbox}{width=1\textwidth,height=0.55\textwidth}
				\includegraphics[width=\textwidth]{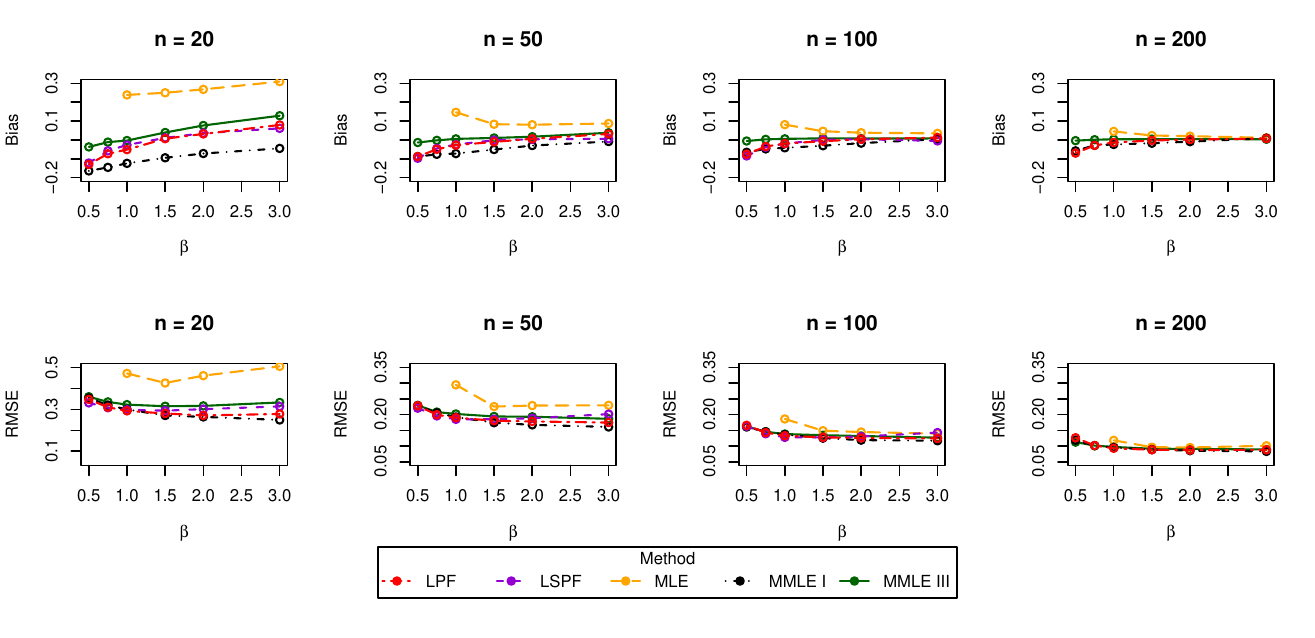}
			\end{adjustbox}
		\end{shaded}
		\caption{Plots for bias and RMSE of the estimator of scale parameter based on various estimation methods varying $\beta$ values.}
		\label{fig_2}
	\end{figure}
		\begin{figure}
			\centering
			\begin{shaded}
				\begin{adjustbox}{width=1\textwidth,height=0.55\textwidth}
					\includegraphics[width=\textwidth]{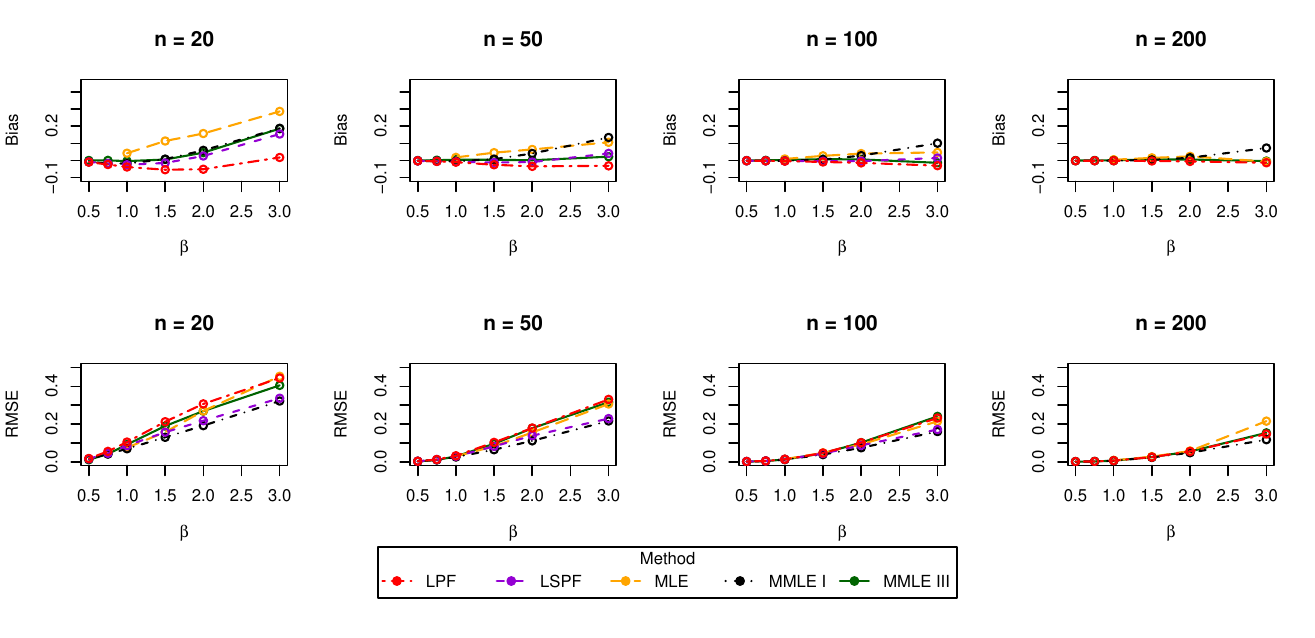}
				\end{adjustbox}
			\end{shaded}
			\caption{Plots for bias and RMSE of the estimator of location parameter based on various estimation methods varying $\beta$ values.}
			\label{fig_3}
		\end{figure}
	\begin{itemize}
		\item As the value of the shape parameter $\beta$ approaches zero or as the sample grows in size, the performance of each method for estimating all parameters improves.
		\item MMLE I: When $\beta\leq 1$, it underestimates the scale parameter, and when $\beta > 1$, it overestimates the scale parameter. Also, the shape parameter is underestimated when $\beta > 1$.
		\item LSPF and LPF methods are nearly always similar in performance when estimating the shape and scale parameters.
		\begin{itemize}
			\item To estimate the shape parameter:
			\begin{itemize}
				\setlength{\topsep}{0em}
				\item If $\beta > 1$, it is recommended to use MMLE III for $n \leq 20$, MLE for $20 < n \leq 100$, and either LPF or MMLE III for $n > 100$.
				\item If $ \beta \leq 1$, it is best advised to use MMLE III.
			\end{itemize}
			\item To estimate the scale parameter:
			\begin{itemize}
				\item When $ \beta > 1$, LPF is suggested as the first choice due to its small bias and RMSE, followed by MMLE III as the second choice because it performs slightly worse than LPF.
				\item If $\beta \leq 1$, it is best advised to use MMLE III.
				\item For small sample sizes, such as $n \leq 20$, we strongly discourage the use of MLE. In this situation, LPF or MMLE I is preferable if RMSE is the prime concern while MMLE III is preferable if bias is the prime concern.
				\item LSPF and LPF perform similar and outperform all other methods.
			\end{itemize}
		\end{itemize}
		\item To estimate the location parameter:
		\begin{itemize}
			\item LPF and LSPF perform similar in terms of biases and RMSEs for $\beta \leq 2$. As $\beta$ moves away from $2$ when $n \leq 50$, LPF method surpasses the LSPF method in terms of bias whereas LSPF outperforms LPF in terms of RMSE. The explanation for this might be that the LSPF produces estimates individually and sequentially, using the estimates obtained in the previous sequence, resulting in a substantial accumulation of biases in the sequences. This is one disadvantage of the LSPF approach over the LPF method.
			\item LPF consistently performs better than all other methods in terms of RMSE and bias, followed by MMLE III, when estimating location parameters. We do not advise MMLE I.
		\end{itemize}
		\item MMLE I is the only method that estimates parameters with a proportion of rejection of approximately zero, regardless of $\beta$ value and sample size. On the basis of the reported proportion of rejections, we recommend using this method as a first preference. If we only consider the LSPF and LPF methods, LSPF has a lower proportion of rejection than LPF. However, if we compare LPF to all other methods including LSPF, we notice the following: After MMLE I, the second preference is LPF for $ 0.5 \leq \beta < 3$ and MLE for relatively large values of $(\beta \geq 3)$; when $\beta$ is very small $(\beta \approx 0.5)$, it has been observed that all the methods have a very small proportion of rejection ($\approx 0$) for sample size $n \geq 50$ and therefore any method can be recommended but when $n \leq 20$, we advise to use LPF.
	\end{itemize}
	\subsection{Evaluation of Quantile Estimation}
	\label{subsec3.2}
	In practice, it is essential to estimate all three parameters in such a way that their combined performance is satisfactory. One may consider the quantile estimation as a possible solution to this problem, in which all parameters are utilized to estimate a qantile of the distribution. We therefore present quantile estimates for $\zeta = 0.01, ~0.05, ~0.10, ~0.25, ~0.50, ~ 0.75, ~ 0.90, ~ 0.95, ~ 0.99$ using equation \eqref{eq11} based on Monte Carlo simulation for the same values of shape parameter and sample size as in Subsection \ref{subsec3.1}. 
	
	\begin{figure}
		\begin{shaded}
			\begin{subfigure}[b]{1\textwidth}
				\begin{adjustbox}{width=1\textwidth,height=0.42\textwidth}
					\includegraphics[width=\textwidth]{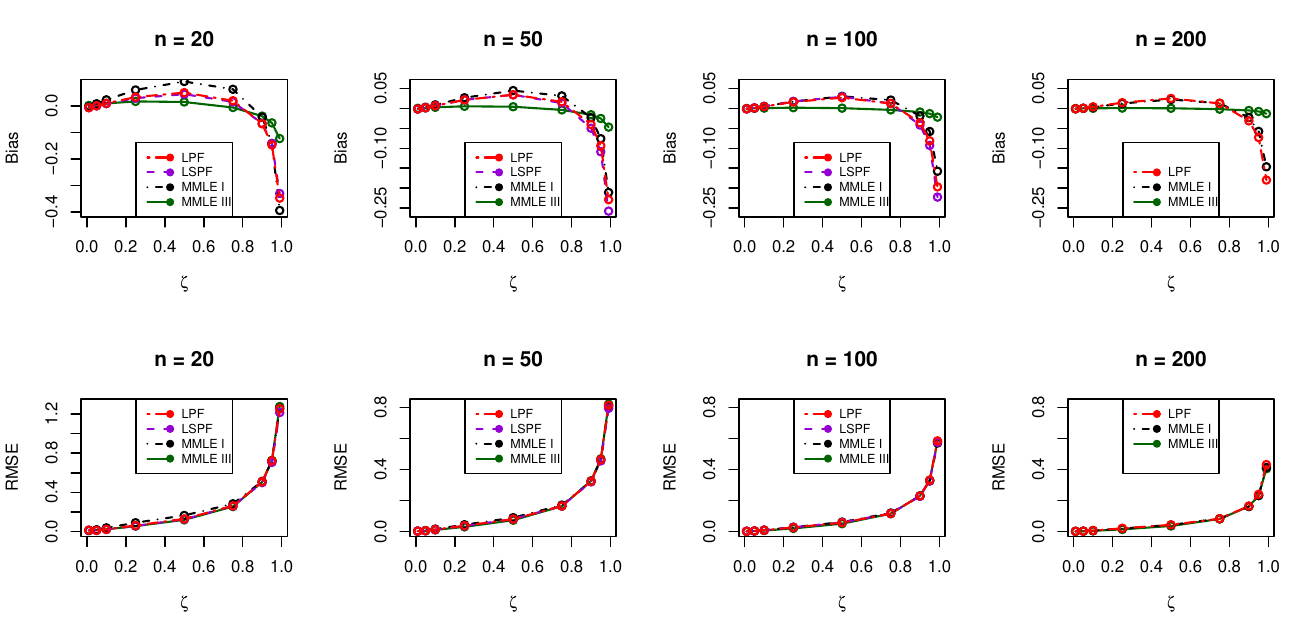}
				\end{adjustbox}
				\caption{$\beta=0.5$}
			\end{subfigure} \\
			\begin{subfigure}[b]{1\textwidth}
				\begin{adjustbox}{width=1\textwidth,height=0.42\textwidth}
					\includegraphics[width=\textwidth]{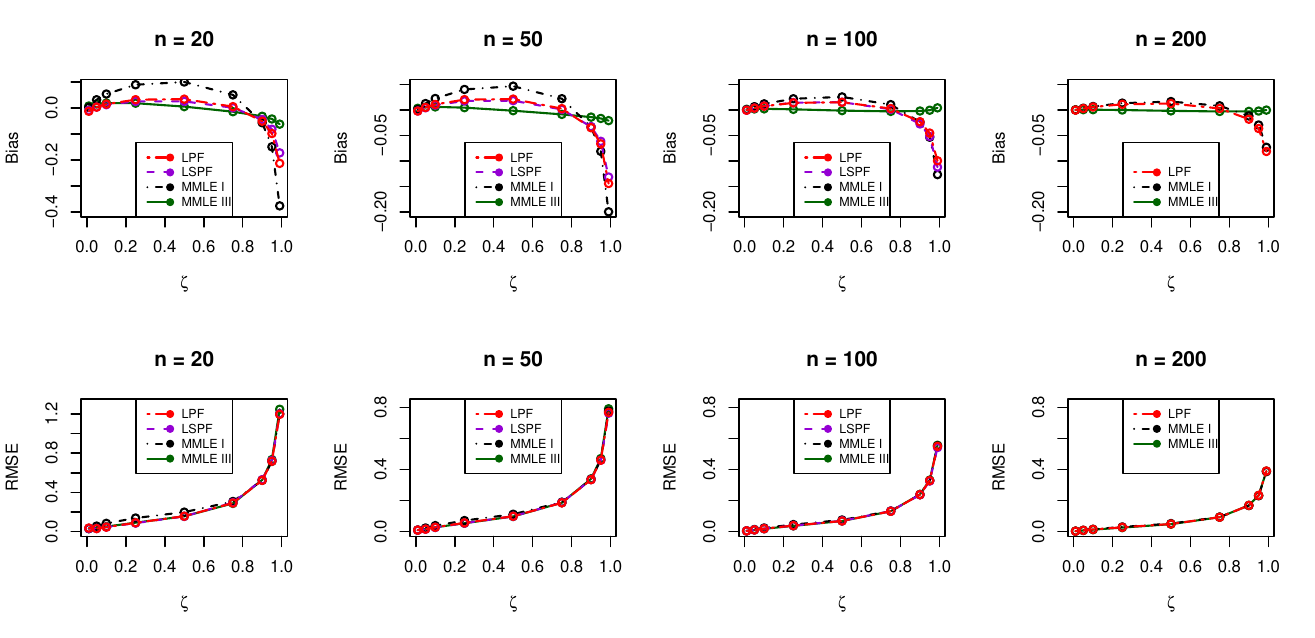}
				\end{adjustbox}
				\caption{$\beta=0.75$}
			\end{subfigure} \\
			\begin{subfigure}[b]{1\textwidth}
				\begin{adjustbox}{width=1\textwidth,height=0.42\textwidth}
					\includegraphics[width=\textwidth]{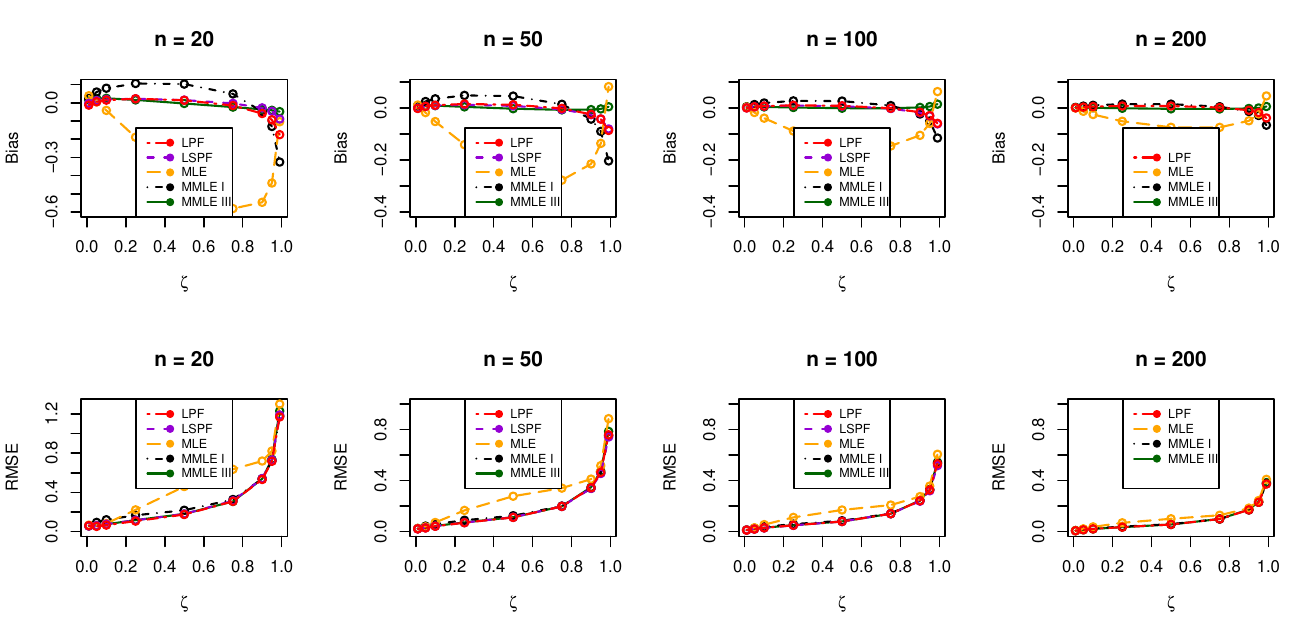}
				\end{adjustbox}
				\caption{$\beta=1.00$}
			\end{subfigure}
		\end{shaded}
		\caption{Plots for bias and RMSE based on various estimation methods when $\beta=0.50, ~ 0.75, ~ 1.00$.}
		\label{fig_4}
	\end{figure}
	\begin{figure}
		\begin{shaded}
			\begin{subfigure}[b]{1\textwidth}
				\begin{adjustbox}{width=1\textwidth,height=0.42\textwidth}
					\includegraphics[width=\textwidth]{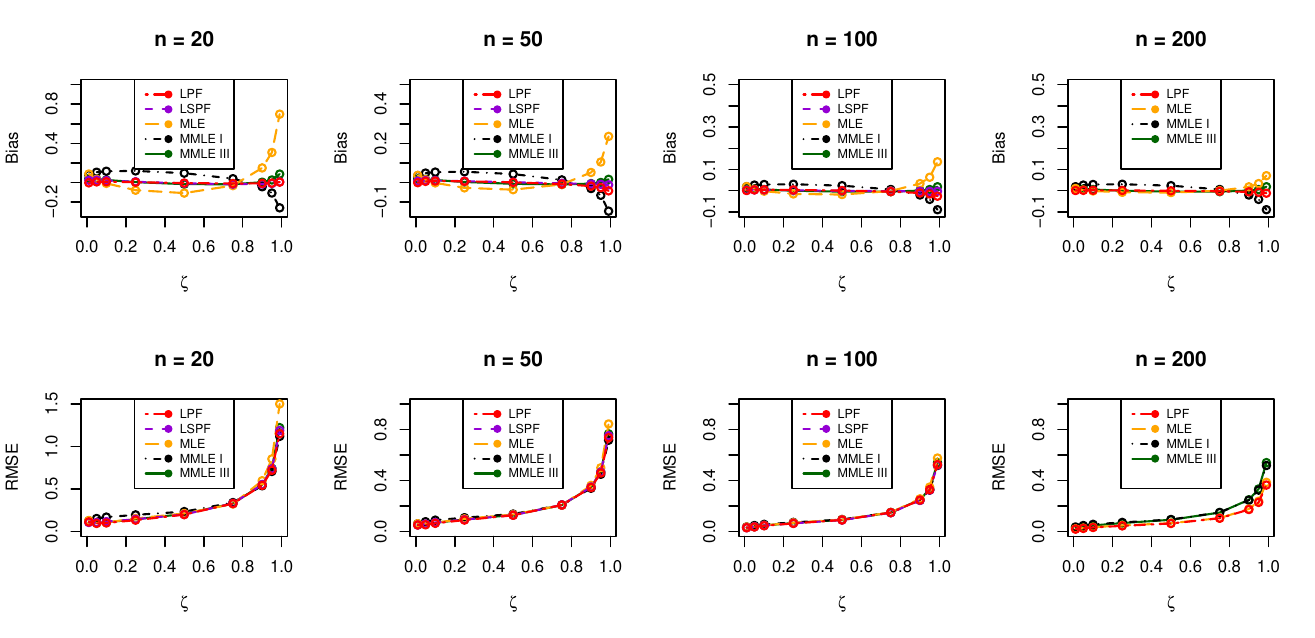}
				\end{adjustbox}
				\caption{$\beta=1.5$}
			\end{subfigure} \\
			\begin{subfigure}[b]{1\textwidth}
				\begin{adjustbox}{width=1\textwidth,height=0.42\textwidth}
					\includegraphics[width=\textwidth]{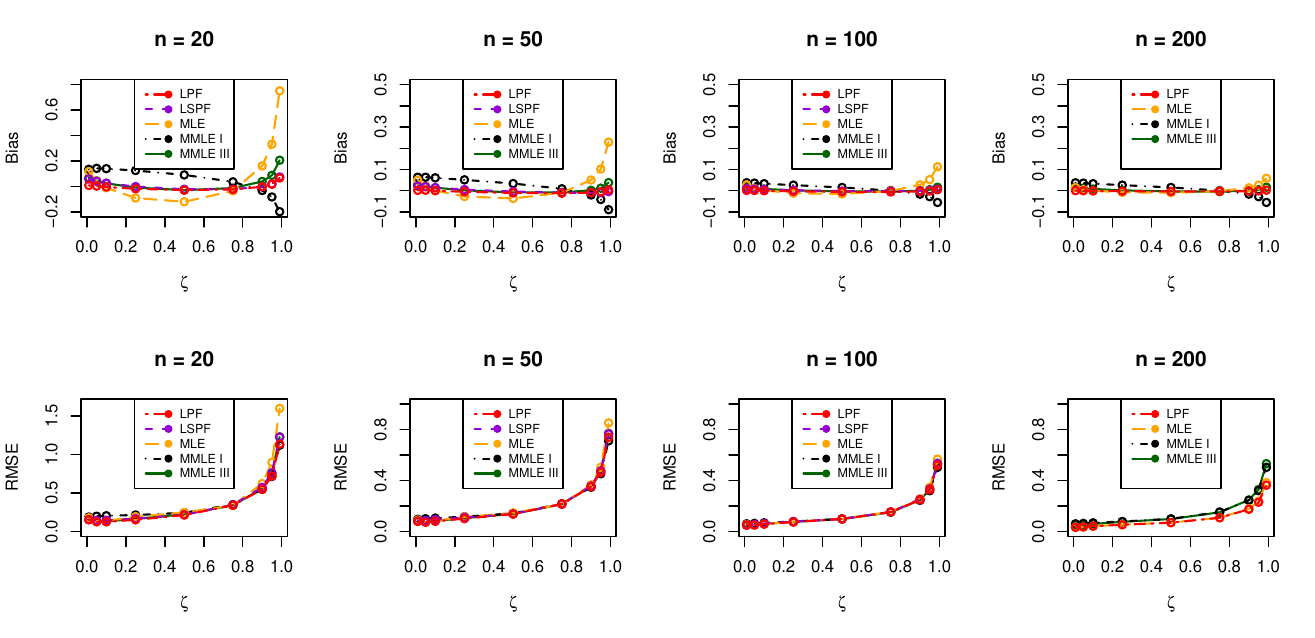}
				\end{adjustbox}
				\caption{ $\beta=2.0$}
			\end{subfigure} \\
			\begin{subfigure}[b]{1\textwidth}
				\begin{adjustbox}{width=1\textwidth,height=0.42\textwidth}
					\includegraphics[width=\textwidth]{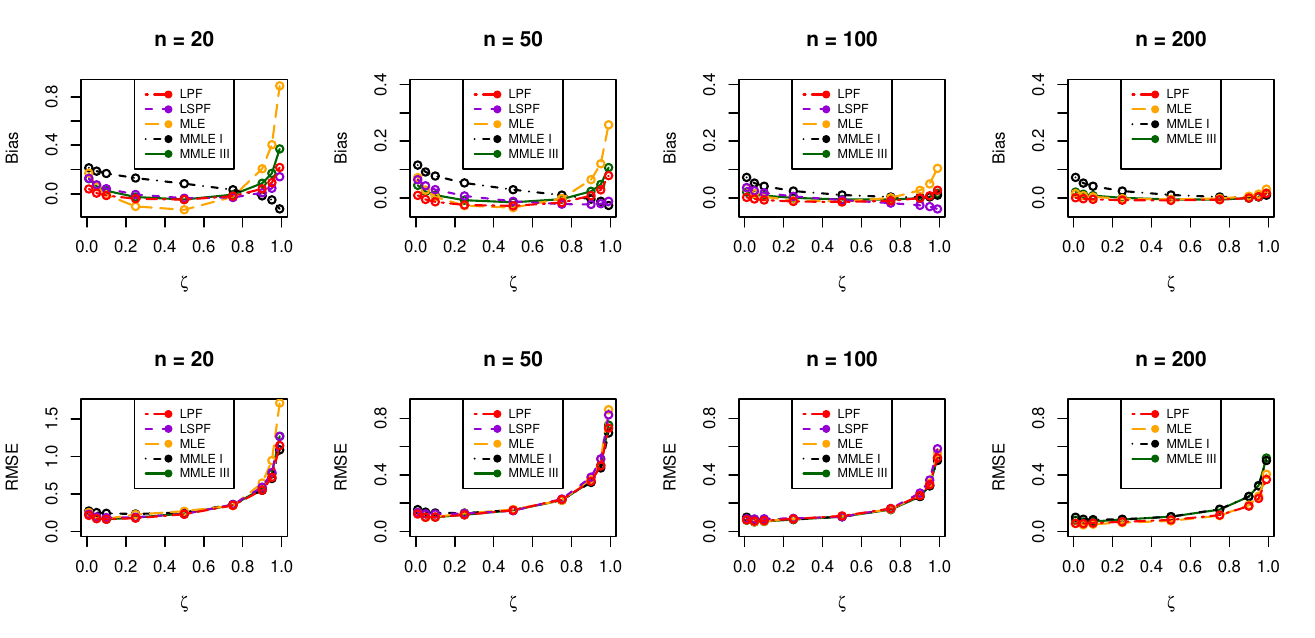}
				\end{adjustbox}
				\caption{$\beta=3.0$}
			\end{subfigure}
		\end{shaded}
		\caption{Plots for bias and RMSE based on various estimation methods when $\beta=1.5, ~ 2.0, ~ 3.0.$}
		\label{fig_5}
	\end{figure}
	In Tables \ref{table3}-\ref{table6}, the biases and RMSEs of the quantile estimation are shown. Figures \ref{fig_4}-\ref{fig_5} exhibit plots of the biases and RMSEs related to each method as the parameter value $\beta$ is varied for each of the sample sizes previously studied. The following can be deduced from these results:
	\begin{itemize}
		\item The absolute value of the bias grows as $\zeta$ approaches towards $0 ~\text{or} ~ 1$ from $0.5$, whereas the RMSE increases as $\zeta$ approaches to $1$ from $0$.
		\item All methods are equivalent in terms of RMSE, particularly when $n > 20$. 
		\item LPF and LSPF comparisons: LPF and LSPF methods yield similar performances almost everywhere except for the situation when $\beta > 2$ and $n <100$. For this exception, we recommend LPF when $\zeta$ is small $(<0.2)$, LSPF when $\zeta$ is close to $1$ $(>0.8)$ and any of the LPF or LSPF when $0.2 \leq \zeta \leq 0.8$.
		\item In terms of bias and RMSE, LPF performs much better than all other methods when $ \beta \geq 1$. It has the lowest bias values of the quantile estimates. Consequently, LPF is strongly advised when $ \beta \geq 1$.
		\item When $ \beta \leq 1$, MMLE III provides the highest performance for any choice of $0 < \zeta < 1$, although LPF performs very well for $0 < \zeta < 0.7$ and much better as $n$ grows. In this instance, we advise LPF for mainly sample sizes $n >20$.
	\end{itemize}
	\section{Illustrative Example}
	\label{section4}
	
The GE distribution is generally used for modeling a wide range of real-world skewed datasets. For instance, such data can be originated from the field of medical sciences, reliability engineering, insurance, and economics. In this section, we consider an electrical failure dataset and illustrate its analysis using the three-parameter GE distribution. For electrical lifetime data, the location parameter and higher quantiles are of particular importance in engineering applications.
	
We analyze a real dataset initially reported by \cite{bain1987introduction} (see example $4.6.3$ on page 162) and subsequently analyzed by \cite{prajapat2021consistent}. The data indicate the observed lifetimes in months of a random sample of $40$ electrical components. Histogram plot of the dataset reveals that the shape of the density is reversed `J' shaped and a simple data analysis provides its summary in Table \ref{table_DA_E2}.

Based on the simulation study, for this type of electrical lifetime data, where the location parameter and higher quantiles are most relevant for engineering applications, either LPF or LSPF is recommended for estimating the location parameter when $\beta \leq 2$. For high-quantile estimation, the simulation results indicate that LPF generally yields lower bias and RMSE when $\beta \geq 1$. In this example, our aim is simply to illustrate how the different methods behave on a practical dataset. We now discuss estimates for the model parameters presented in the Table \ref{table_DA_E2}. 

\begin{table}[ht]
	\centering
	\caption{Electrical lifetime data, its summary and estimates of the parameters.}
	\begin{tabular}{lccccccc}
		\toprule
		\multicolumn{8}{c}{Data:}\\
		\multicolumn{8}{c}{0.15, 2.37, 2.90, 7.39, 7.99, 12.05, 15.17, 17.56, 22.40, 34.84, 35.39,  36.38, 39.52, 41.07, 46.50, } \\
		\multicolumn{8}{c}{  50.52,	52.54, 58.91, 58.93, 66.71, 71.48, 71.84, 77.66, 79.31, 80.90, 90.87, 91.22, 96.35, } \\
		\multicolumn{8}{c}{ 108.92, 112.26, 122.71, 126.87, 127.05, 137.96, 167.59, 183.53, 282.49, 335.33, 341.19, 409.97}  \\
		\midrule
		\multicolumn{8}{c}{Summary:}\\
		\multicolumn{8}{c}{~ Min. ~~ 1st Qu. ~ Median ~ Mean ~ 3rd Qu. ~ Max. ~~~ Skew ~ Kurtosis} \\
		\multicolumn{8}{c}{~ 0.15 ~~~~~ 35.25 ~~~ 69.09 ~~~~ 93.12 ~~ 114.87 ~~~ 409.97 ~~ 1.74 ~~~~~ 2.49 ~~ } \\
		\midrule
		\multirow{2}{*}{Method} & \multicolumn{3}{c}{Estimates} & \multicolumn{2}{c}{KS} & \multicolumn{2}{c}{CvM} \\ \cmidrule(r){2-4} 
		\cmidrule(r){5-6} \cmidrule(r){7-8}
		& Shape &Scale & Location & Statistic & p-value & Statistic & p-value \\ \midrule
		LPF & 1.0821 & 91.1620  & -2.7991 &  0.0836 & 0.9323 & 0.0414 & 0.9258 \\
		LSPF & 1.0799 & 91.2659  & -2.7673 & 0.0838 & 0.9311 & 0.0415 & 0.9252  \\
		MMLE I  &  1.2564	& 85.0422 & -2.7578 & 0.0897 & 0.8899 & 0.0543 & 0.8495  \\  
		MMLE III  &  1.0408 & 92.1752 & -1.4577 &  0.0851 & 0.9231 &  0.0419 & 0.9229 \\
		\bottomrule
	\end{tabular}
	\label{table_DA_E2}
\end{table}

The parameter estimates for all methods considered in this study are presented in Table \ref{table_DA_E2}, together with the KS distance statistic, the CvM test statistic, and their associated $p$-values. To demonstrate the LPF approach, we maximize the likelihood function $\ell_v(\alpha, \beta)$ in equation \eqref{eq4} with respect to the parameters $\alpha$ and $\beta$. We estimate $\alpha$ and $\beta$ to be $91.1620$ and $1.0821$, respectively. By following the steps described in Section \ref{sec2}, we estimate $\gamma$ to be $-2.7991$ using the estimators of $\alpha$ and $\beta$. According to the results reported in Table \ref{table_DA_E2} based on the CvM test, LPF, LSPF and MMLE III perform really good for this particular dataset, while LPF performs best. In Table \ref{table_DA_Qs_E2}, quantiles estimates are reported for this particular dataset.

\begin{figure}[ht]
	\centering
	\subcaptionbox{\label{subfig_61}}%
	{\includegraphics[width=0.49\linewidth, height=0.4\linewidth]{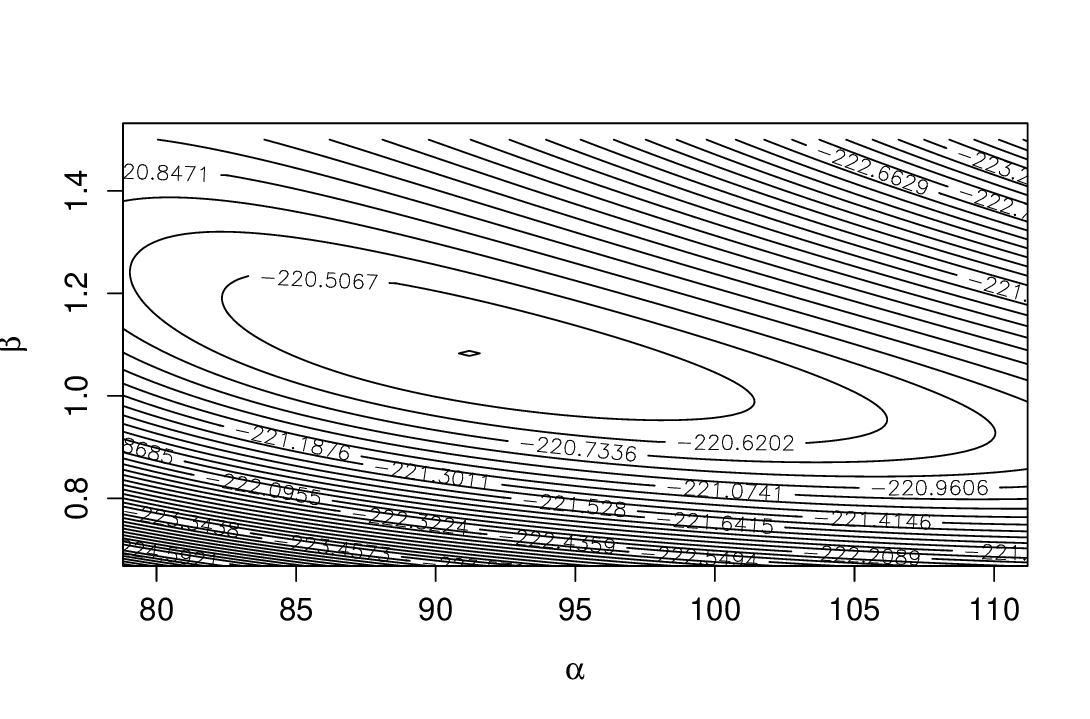}}
	\hfill
	\subcaptionbox{\label{subfig_62}}%
	{\includegraphics[width=0.49\linewidth, height=0.4\linewidth]{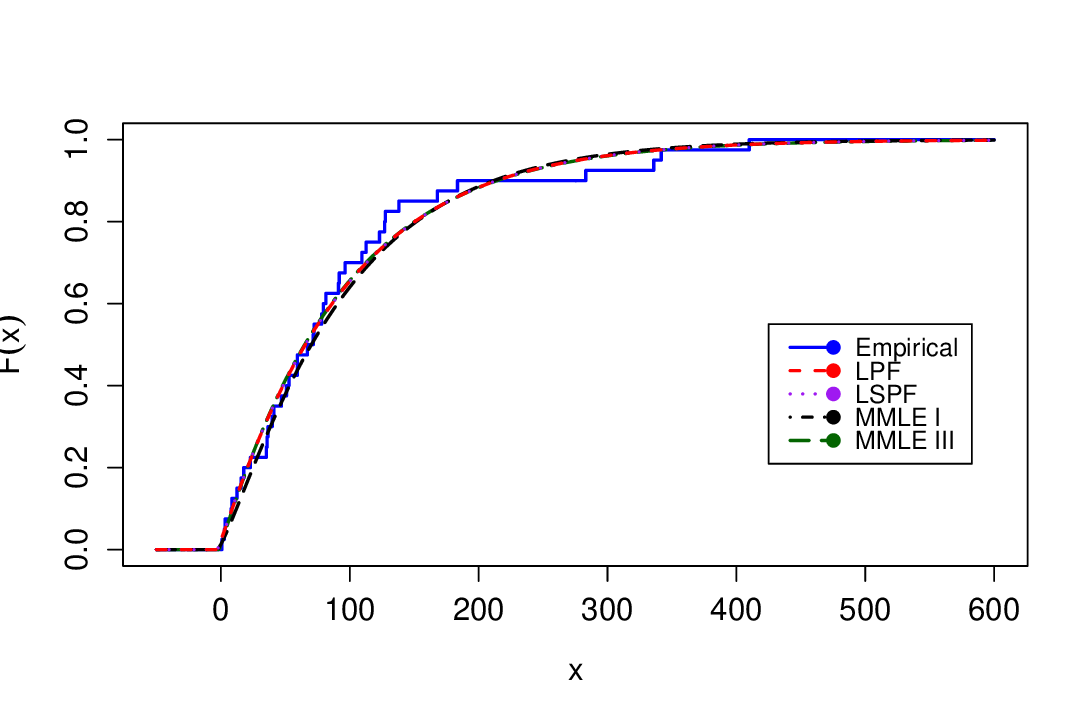}}
	\caption{(a) Plot of the log-likelihood function $\ln (\ell_v(\alpha, \beta))$; (b) Fitted CDF plots along with the empirical CDF.}
	\label{fig_6}
\end{figure}
\begin{table}[ht]
	\caption{Various quantile estimates for the lifetime model of the electrical components under the illustrative example.}
	\centering
	\small
	\begin{tabular}{cccccccccc}
		\toprule
		\multirow{2}{*}{Method} & \multicolumn{8}{c}{$\zeta$} &    \\ \cmidrule{2-10}
		&  0.01 & 0.05 & 0.10 & 0.25 & 0.50 & 0.75 & 0.90 & 0.95 & 0.99  \\ \midrule
		LPF  & -1.4970 &   3.1096 &   8.7601 &  26.8607 &  65.4500 & 129.8222 & 213.9440 & 277.3149 & 424.1757 \\
		LSPF & -1.4750 &   3.1138  &  8.7517  & 26.8349  & 65.4270 & 129.8431 & 214.0460 & 277.4846 & 424.5093 \\
		MMLE I  & -0.5531  & 5.4629 & 12.0672 & 31.5204 & 70.2034 & 132.1558 & 211.5682 & 270.9731 & 408.1978 \\ 
		MMLE III  & -0.3469    &3.8768  &  9.2265 &  26.7915 &  65.0078  &129.5150 & 214.2825 & 278.2687 & 426.6927 \\
		\bottomrule
	\end{tabular}
	\label{table_DA_Qs_E2}
\end{table}	
\begin{table}[h!]
\begin{center}
	\caption{Bootstrap CIs for the illustrative example based on 10,000 simulations.}
	\begin{tabular}{cccccc}
		\toprule
		CI                    & Method   & Shape & Scale & Location & p \\ \midrule
		\multirow{5}{*}{95\%} 
		& LPF & (0.7102,    2.3575) & ( 56.0280,  132.2849) & ( -14.1289,   4.4647) & 0.0046 \\
		& LSPF     &   (0.7462,	2.3137)
		&  (56.7866,	130.8569)      &     (-12.8434,	4.5948)  & 0.0020 \\
		& MMLE I   &  (0.9125, 2.8808) & (50.5208, 113.5426)  & (-10.9357, 7.4144) & 0.0000 \\ 
		& MMLE III &    (0.6568,  2.1126) & (57.0883,  137.9524) & (-8.8587,  6.7738)	& 0.0014  \\ \midrule
		\multirow{5}{*}{99\%} 
		& LPF  & (0.6400,    3.4396) & ( 48.6631, 149.7184) & (-22.8317,    8.7151) & 0.0046 \\
		& LSPF  &   (0.6765,	3.7231)  & (48.3591,	147.1999)   &   (-19.4118,	8.3977) & 0.0020 \\
		& MMLE I   & (0.8051, 3.7779) &   (44.0821, 129.2398) & (-14.1193, 12.5442) & 0.0000\\ 
		& MMLE III &        (0.5831, 3.1162) & (48.3295,  157.8311)       &  (-17.9056,   10.3762) & 0.0014  \\ \midrule
	\end{tabular}
\label{table_DA_CI_E28}
\end{center}
\end{table}
\begin{figure}[h!]
\centering
	\begin{adjustbox}{width=1\textwidth,height=0.55\textwidth}
		\includegraphics[width=\textwidth]{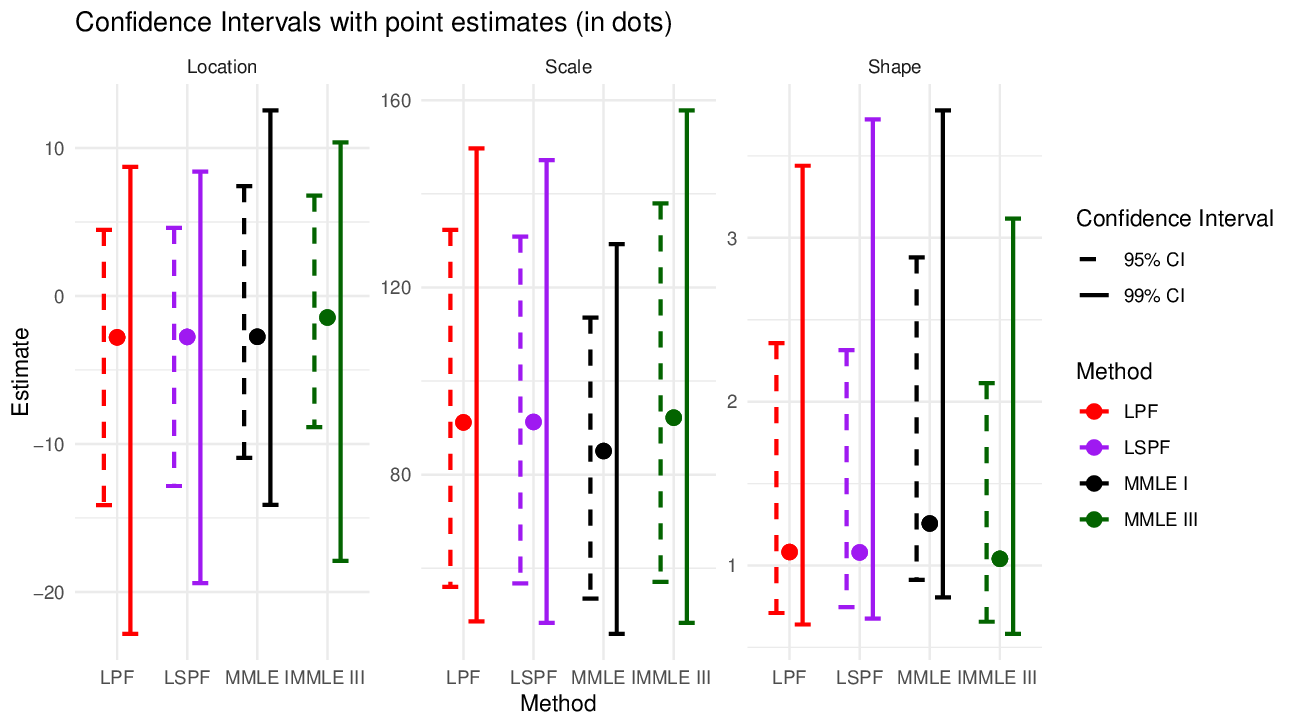}
	\end{adjustbox}
\caption{Graphical representation of the bootstrap CIs along with the point estimates for the illustrative example.}
\label{fig_9}
\end{figure} 

Subfigure \ref{subfig_61} shows a graphic that may be used to determine reasonable initial parameter values. It also shows that the likelihood function is unimodal, meaning that the likelihood is globally maximized by the derived estimates. Plots of the empirical CDF and the fitted CDF are shown in Subfigure \ref{subfig_62}. To find the bootstrap confidence intervals (CIs), we set $\beta_U = 12$. Based on this dataset, Table \ref{table_DA_CI_E28} and Figure \ref{fig_9} present $95\%$ and $99\%$ bootstrap CIs for each method where we found that the proportion of rejection is negligible for all the methods. For this data analysis, each method’s CI contains its corresponding point estimate. Overall, the results for this dataset agree with the simulation study: LPF and LSPF give stable parameter and quantile estimates. All methods provide comparable fits based on the CDF and goodness-of-fit tests. It is to be noted that, the location parameter $\gamma$ under this study belongs to the set of real numbers $R$, hence the estimation method and the algorithms are proposed accordingly. Although, the estimation of the location parameter can be constrained to $\gamma \ge 0$ if the application requires a physically interpretable location parameter. 
\section{Summary}
\label{section5}
In this paper, we consider an earlier-proposed estimation method for a three-parameter GE distribution known as the LPF method. We discussed some properties of the estimators, such as uniqueness and consistency. A Monte Carlo simulation study has been conducted to evaluate the performance of the LPF method comparative to other existing methods. In numerical simulations, we reported bias and RMSE for estimators of all three parameters and quantiles of the GE distribution. For $\beta > 1$, the LPF method performs better than all other methods for estimating scale parameters in terms of bias and RMSE. In addition to other methods such as MMLE I, MLE, etc., MMLE III has excellent and consistent performance when estimating shape parameter. When $\beta <1$, the LPF method is recommended for the estimation of the shape parameter when the sample size is small. When $n \geq 50$, MMLE III is advisable if $\beta$ is sufficiently small $(\beta \approx 0.5)$ and any method is acceptable if $0.5 < \beta \leq 1$ as they all have similar performances for all parameters. The LPF method is recommended for estimating quantiles, especially when $\beta > 1$.  Quantile estimation based on the LPF method appeared performing very good even for large values of the shape parameter. Based on the reported results of the proportion of rejected samples during the simulation study, it is recommended that LSPF and MMLE I be used for data analysis before any other methods. A real dataset from reliability engineering, has been used to illustrate the LPF method. Moreover, the proposed method performs very good among existing methods for the presented dataset. 

According to this study, the LPF method is superior to the LSPF method in a number of aspects and offers certain advantages: (1) It consumes less time to run simulations and has a lower computational complexity. The primary cause for this is the decrease in the number of integration; (2) One disadvantage of the LSPF approach over the LPF method is that it employs estimates from the previous step in order to generate estimates separately in steps. This might result in a considerable accumulation of biases in the steps. The LPF, on the other hand, creates estimates by reducing the number of steps in the estimating approach; (3) the LPF  method is based on $n-1$ observations, whereas the LSPF method makes use of just $n-2$ observations.
\section*{Supplementary Material}
The supplementary material contains a ZIP archive (\texttt{Rcodes.zip}) including all \texttt{R} scripts used for the Monte Carlo simulations and the illustrative data analyses. Detailed instructions and folder descriptions are provided in an accompanying README file.
\section*{Acknowledgments}
The authors thank the anonymous reviewers for their constructive feedback, which has substantially improved the quality of the manuscript. The authors further acknowledge the support of the High-Performance Computing (HPC) systems at the Indian Institute of Technology Kanpur and Newcastle University.
\section*{Conflict of Interest}
We hereby declare that the information provided here is accurate, and no apparent conflicts of interest that are relevant to the content of this article. 
\bibliographystyle{apalike}
\bibliography{references}

@article{mudholkar1995exponentiated,
	title={The exponentiated {W}eibull family: A reanalysis of the bus-motor-failure data},
	author={Mudholkar, Govind S and Srivastava, Deo Kumar and Freimer, Marshall},
	journal={Technometrics},
	volume={37},
	number={4},
	pages={436--445},
	year={1995},
	publisher={Taylor \& Francis}
}

@article{guptakundu1999,
	title={Theory \& methods: {G}eneralized exponential distributions},
	author={Gupta, Rameshwar D and Kundu, Debasis},
	journal={Australian \& New Zealand Journal of Statistics},
	volume={41},
	number={2},
	pages={173--188},
	year={1999},
	publisher={Wiley Online Library}
}

@article{gupta2001exponentiated,
	title={Exponentiated exponential family: an alternative to gamma and {W}eibull distributions},
	author={Gupta, Rameshwar D and Kundu, Debasis},
	journal={Biometrical Journal: Journal of Mathematical Methods in Biosciences},
	volume={43},
	number={1},
	pages={117--130},
	year={2001},
	publisher={Wiley Online Library}
}

@article{gupta2007generalized,
	title={Generalized exponential distribution: {E}xisting results and some recent developments},
	author={Gupta, Rameshwar D and Kundu, Debasis},
	journal={Journal of Statistical Planning and Inference},
	volume={137},
	number={11},
	pages={3537--3547},
	year={2007},
	publisher={Elsevier}
}

@article{raqab2008estimation,
	title={Estimation of {P(Y}$<${X)} for the three-parameter generalized exponential distribution},
	author={Raqab, Mohammad Z and Madi, Mohamed T and Kundu, Debasis},
	journal={Communications in Statistics—Theory and Methods},
	volume={37},
	number={18},
	pages={2854--2864},
	year={2008},
	publisher={Taylor \& Francis}
}

@article{nagatsuka2013consistenta,
	title={A consistent method of estimation for the three-parameter {W}eibull distribution},
	author={Nagatsuka, Hideki and Kamakura, Toshinari and Balakrishnan, N},
	journal={Computational Statistics \& Data Analysis},
	volume={58},
	pages={210--226},
	year={2013},
	publisher={Elsevier}
}

@article{nagatsuka2014consistent,
	title={A consistent method of estimation for the three-parameter gamma distribution},
	author={Nagatsuka, Hideki and Balakrishnan, N and Kamakura, Toshinari},
	journal={Communications in Statistics-Theory and Methods},
	volume={43},
	number={18},
	pages={3905--3926},
	year={2014},
	publisher={Taylor \& Francis}
}

@article{nagatsuka2012parameter,
	title={Parameter and quantile estimation for the three-parameter gamma distribution based on statistics invariant to unknown location},
	author={Nagatsuka, Hideki and Balakrishnan, N},
	journal={Journal of Statistical Planning and Inference},
	volume={142},
	number={7},
	pages={2087--2102},
	year={2012},
	publisher={Elsevier}
}

@article{cohen1982modified,
	title={Modified maximum likelihood and modified moment estimators for the three-parameter {W}eibull distribution},
	author={Cohen, Clifford A and Whitten, Betty},
	journal={Communications in Statistics-Theory and Methods},
	volume={11},
	number={23},
	pages={2631--2656},
	year={1982},
	publisher={Taylor \& Francis}
}

@article{cohen1980estimation,
	title={Estimation in the three-parameter lognormal distribution},
	author={Cohen, A Clifford and Whitten, Betty Jones},
	journal={Journal of the American Statistical Association},
	volume={75},
	number={370},
	pages={399--404},
	year={1980},
	publisher={Taylor \& Francis}
}

@article{cohen1984modified,
	title={Modified moment estimation for the three-parameter {W}eibull distribution},
	author={Cohen, A Clifford and Whitten, Betty Jones and Ding, Yihua},
	journal={Journal of Quality Technology},
	volume={16},
	number={3},
	pages={159--167},
	year={1984},
	publisher={Taylor \& Francis}
}

@article{pasari2014three,
	title={Three-parameter generalized exponential distribution in earthquake recurrence interval estimation},
	author={Pasari, Sumanta and Dikshit, Onkar},
	journal={Natural hazards},
	volume={73},
	number={2},
	pages={639--656},
	year={2014},
	publisher={Springer}
}

@article{hall2005bayesian,
	title={{B}ayesian likelihood methods for estimating the end point of a distribution},
	author={Hall, Peter and Wang, Julian Z},
	journal={Journal of the Royal Statistical Society: Series B (Statistical Methodology)},
	volume={67},
	number={5},
	pages={717--729},
	year={2005},
	publisher={Wiley Online Library}
}

@book{billingsley2008probability,
	title={Probability and measure},
	author={Billingsley, Patrick},
	year={1994},
	publisher={John Wiley \& Sons (3rd ed.)},
	address={New York}
}

@article{ghitany2013,
	title={On the existence and uniqueness of the {MLE}s of the parameters of a general class of exponentiated distributions},
	author={Ghitany, ME and Al-Jarallah, RA and Balakrishnan, N},
	journal={Statistics},
	volume={47},
	number={3},
	pages={605--612},
	year={2013},
	publisher={Taylor \& Francis}
}

@book{bain1987introduction,
	title={Introduction to probability and mathematical statistics},
	author={Bain, Lee J and Engelhardt, Max},
	year={1988},
	publisher={Duxbury Press},
	address={Boston}
}

@article{nagatsuka2015efficient,
	title={An efficient method of parameter and quantile estimation for the three-parameter {W}eibull distribution based on statistics invariant to unknown location parameter},
	author={Nagatsuka, Hideki and Balakrishnan, N},
	journal={Communications in Statistics-Simulation and Computation},
	volume={44},
	number={2},
	pages={295--318},
	year={2015},
	publisher={Taylor \& Francis}
}

@article{prajapat2021consistent,
  title={A consistent method of estimation for three-parameter generalized exponential distribution},
  author={Prajapat, Kiran and Mitra, Sharmishtha and Kundu, Debasis},
  journal={Communications in Statistics-Simulation and Computation},
  volume={52},
  number={6},
  pages={2471--2487},
  year={2021},
  publisher={Taylor \& Francis}
}

@article{basu2023three,
  title={On three-parameter generalized exponential distribution},
  author={Basu, Suparna and Kundu, Debasis},
  journal={Communications in Statistics-Simulation and Computation},
  volume={53},
  number={12},
  pages={6222--6236},
  year={2023},
  publisher={Taylor \& Francis}
}

@book{shao2003mathematical,
	title        = {Mathematical Statistics},
	author       = {Shao, Jun},
	edition      = {2nd},
	year         = {2003},
	publisher    = {Springer},
	address      = {New York},
	series       = {Springer Texts in Statistics},
	isbn         = {978-0-387-98774-2}
}

@article{gupta2001generalized,
	title={Generalized exponential distribution: different method of estimations},
	author={Gupta, Rameshwar D and Kundu, Debasis},
	journal={Journal of statistical computation and simulation},
	volume={69},
	number={4},
	pages={315--337},
	year={2001},
	publisher={Taylor \& Francis}
}

@article{smith1985maximum,
	title={Maximum likelihood estimation in a class of nonregular cases},
	author={Smith, Richard L},
	journal={Biometrika},
	volume={72},
	number={1},
	pages={67--90},
	year={1985},
	publisher={Oxford University Press}
}

@article{smith1987comparison,
	title={A comparison of maximum likelihood and Bayesian estimators for the three-parameter Weibull distribution},
	author={Smith, Richard L and Naylor, JC918854},
	journal={Journal of the Royal Statistical Society Series C: Applied Statistics},
	volume={36},
	number={3},
	pages={358--369},
	year={1987},
	publisher={Oxford University Press}
}

@article{hirose1997inference,
	title={Inference from grouped data in three-parameter Weibull models with applications to breakdown-voltage experiments},
	author={Hirose, Hideo and Lai, Tze Leung},
	journal={Technometrics},
	volume={39},
	number={2},
	pages={199--210},
	year={1997},
	publisher={Taylor \& Francis}
}
\section*{Ethics declarations}
Funding information - not applicable.
\appendix
\section{Proof of Theorem \ref{thm_2.1}}
\label{app_1}
Using the transformation of random variables, we find the joint PDF of the random variables $V_{(2)},V_{(3)}, \dots, V_{(n)}$ for $ \alpha > 0, ~\beta>0 $. Define $Z_{(i)} = {X_{(i)} - \gamma}, i = 1,2, \dots,n$. $Z_1, Z_2, \dots,Z_n$  are \textit{i.i.d.} random variables from the GE distribution $\text{GE}(\alpha,\beta,0)$ with the shape parameter $\beta$ and the scale parameter $\alpha$. Denote the PDF and the CDF of $\text{GE}(\alpha,\beta,0)$ by $g(.;\alpha,\beta)$ and $G(.;\alpha, \beta)$, respectively, for convenience. Now, $V_{(i)} $ in equation $(\refeq{eq3})$ can be rewritten in terms of $Z_{(i)}$'s as follows:
\begin{align*}
& V_{(i)} = {Z_{(i)} - Z_{(1)}}, ~ i=2, 3, \dots,n  \implies Z_{(i)} = Z_{(1)} + V_{(i)},~ i=2, 3, \dots,n.&.
\end{align*}   
Let us assume that $ U = Z_{(1)}$. Therefore $ Z_{(i)} = U + V_{(i)},~ i=2, 3, \dots,n.$
It can be shown that the Jacobian of the transformation $ J = \frac{\partial (Z_{(1)}, Z_{(2)}, \dots, Z_{(n-1)},Z_{(n )}) }{ \partial (U,V_{(2)}, \dots, V_{(n-1)},V_{(n) })} = 1 $. If we use a notation $f_{Y_1,Y_2,\dots,Y_p}(.)$ to denote the joint PDF of $Y_1,Y_2,\dots,Y_p$, we have
\begin{align*}
f_{U,V_{(2)}, \dots, V_{(n)}}(u, v_2, \dots, v_n;\alpha,\beta) 	&  = |J| f_{Z_{(1)},Z_{(2)},\dots,Z_{(n-1)},Z_{(n)}}(u, u + v_2,\dots,u + v_n;\alpha,\beta ) & \\
&  = n! g(u;\alpha, \beta) \Big\{ \prod_{i = 2}^{n} g(u+v_i;\alpha, \beta) \Big\},  &
\end{align*}
$0 < v_2 < \dots < v_{n} < \infty$ and $0 < u < \infty$. 
Therefore
\begin{align}
\label{eq5}
& f_{V_{(2)}, \dots, V_{(n)}}(v_2, \dots, v_{n};\alpha,\beta ) &  \nonumber\\ 
& \qquad = \int_{0}^{\infty} n! g(u;\alpha, \beta) \Big\{ \prod_{i = 2}^{n} g(u+v_i;\alpha, \beta) \Big\} \ du  & \\
& \qquad = n! \Big(\frac{\beta}{\alpha}\Big)^n \int_{0}^{\infty}  e^{ -\frac{1}{\alpha}\sum_{i=1}^{n} (u+v_i) } \prod_{i = 1}^{n} \big( 1 - e^{-\frac{u+v_i}{\alpha}} \big)^{\beta-1} \ du, & \nonumber
\end{align} with $v_1=0$ and hence, the likelihood function of $(\alpha, \beta)$ given $v_2, v_3,\dots, v_{n} $ is given by
\begin{align*}
\ell_v(\alpha, \beta|v_2, \dots, v_{n}) & = f_{V_{(2)}, \dots, V_{(n)}}(v_2, \dots, v_{n};\alpha,\beta ),~ \alpha >0, ~ \beta>0 &
\end{align*}
which proves the theorem.
\section{Boundedness and Differentiablity of the Likelihood Function \boldmath{$\ell_v(\alpha,\beta)$}}
\label{app_2}
\subsection{Boundedness}
By equation $\eqref{eq5}$, we have 
\begin{align}
\label{eq6}
\ell_v(\alpha,\beta) & = n! \int_{0}^{\infty} g(u;\alpha,\beta)  \prod_{i = 2}^{n} g(u+v_i;\alpha,\beta) \ du,~ \alpha > 0,~\beta>0. &
\end{align}
It can be shown for $\alpha > 0, ~\beta > 0, \ 0 < v_2 < \dots < v_{n} < \infty$ and $0 < u  < \infty $ that $ \frac{(n-1)! \prod_{i=2}^{n} g( u+v_i;\alpha,\beta)}{(1 - G(u;\alpha, \beta))^{n-1}} $ is bounded , \textit{i.e.}, $ \exists \ $ an $M > 0$ such that $ (n-1)! \prod_{i=2}^{n}  g( u+ v_i;\alpha,\beta) < M (1 - G(u;\alpha, \beta))^{n-1} ~ \forall ~ \alpha>0, ~\beta > 0, \ 0 < v_2 < \dots < v_{n} < \infty$ and $ 0 < u < \infty $. It implies:   
\begin{align*}
\int_{0}^{\infty} n!  g(u;\alpha,\beta) \prod_{i=2}^{n}  g( u+ v_i;\alpha,\beta)  \ du < M \int_{0}^{\infty} n g(u;\alpha,\beta) (1 - G(u;\alpha,\beta))^{n-1} \ du = M,
\end{align*}  
\subsection{Differentiablity}
First we show differentiablity of the likelihood function $\ell_v(\alpha,\beta)$ in equation \eqref{eq4} with respect to $\beta$. In order to do so, let us rewrite the likelihood function as follows:
\begin{align}
\label{eqq3}
\ell_v(\alpha, \beta) = &  n!  \int_{0}^{\infty}  e^{ h_v(\alpha, \beta;u)}  \ du, ~ \alpha > 0,~ \beta > 0, &
\end{align}
where
\begin{align}
\label{eqq4}	
h_v(\alpha, \beta;u) =  n \ln(\beta) - n \ln(\alpha) - \frac{1}{\alpha} \sum_{i=1}^{n} (u+v_i) + (\beta - 1) \sum_{i=1}^{n} \ln(1 - e^{-\frac{u+v_i}{\alpha}}). 
\end{align}
To show that the likelihood function is differentiable, we need to show that the partial derivative, with respect to $\beta$, can be taken inside the integral in the equation $\eqref{eq4}$. Given $0 < v_2 < \dots < v_{n} < \infty, ~ v_{1} = 0 ~ \text{and} ~ \alpha>0$, we show:
\begin{center}
\begin{enumerate}\vspace{-1em}
	\setlength{\itemsep}{0em}
	\item $\frac{\partial}{\partial \beta} e^{h_v(\alpha, \beta;u)} $ exists,
	\item $ \left| \frac{\partial}{\partial \beta} e^{h_v(\alpha,\beta;u)} \right| < h_2{(u)} $ for some positive function $h_2$ and $\forall ~ u \in (0,\infty) $ such that $\int_{0}^{\infty}  h_2(u) \ du < \infty$, \textit{i.e.} $\frac{\partial}{\partial \beta} e^{h_v(\alpha,\beta;u)} $ is an integrable function with respect to the variable $u$.
\end{enumerate}
\end{center}
Since exponential, logarithmic and polynomials are well-known smooth functions, therefore $e^{h_v(\alpha,\beta;u)}$ is differentiable with respect to $\beta$ and it is given by 
\begin{align}
\label{eqq5}
\frac{\partial}{\partial \beta} e^{h_v(\alpha,\beta;u)} = e^{h_v(\alpha,\beta;u)}   \frac{\partial}{\partial \beta} h_v(\alpha,\beta;u)   =   \Big( \frac{n}{\beta} + \sum_{i = 1}^{n} \ln\big(1 - e^{-\frac{u+v_i}{\alpha}}\big) \Big) e^{h_v(\alpha,\beta;u) }.
\end{align} Equation \eqref{eqq5} can be simplified as follows:
\begin{align}
\label{eqq6}
\frac{\partial}{\partial \beta} e^{h_v(\alpha,\beta;u)}  
= \frac{n}{\beta} e^{h_v(\alpha,\beta;u) } + \Big(\frac{\beta}{\alpha}\Big)^n \Big( \sum_{i = 1}^{n} \ln\big(1 - e^{-\frac{u+v_i}{\alpha}}\big) \Big)  \prod_{i = 1}^{n} \Big\{ e^{-\frac{u+v_i}{\alpha}}  \big( 1 - e^{-\frac{u+v_i}{\alpha}} \big)^{\beta-1} \Big\}.
\end{align}
It is clear that first term of the equation \eqref{eqq6} is integrable with respect to the variable $u$. Now, we show integrability of the second term and in order to do so, let us assume that 
\begin{align*}
\mathfrak{H}_1(u) = \Big( \sum_{i = 1}^{n} \ln\big(1 - e^{-\frac{u+v_i}{\alpha}}\big) \Big)  \prod_{i = 1}^{n} \Big\{ e^{-\frac{u+v_i}{\alpha}}  \big( 1 - e^{-\frac{u+v_i}{\alpha}} \big)^{\beta-1} \Big\}.
\end{align*} 
For every $\alpha > 0, ~\beta > 0, ~ 0 < v_2 < \dots < v_{n} < \infty, ~ v_1 = 0 $ and $0 < u < \infty $, $\mathfrak{H}_1(u)$ tends to zero as $u \to \infty$. Also, $\mathfrak{H}_1(u) < h_3(u)$, where $ h_3(u) = C_1 e^{-\frac{u}{\alpha}}  \big( 1 - e^{-\frac{u}{\alpha}} \big)^{\beta-1} \big( C_2 + \ln\big(1 - e^{-\frac{u}{\alpha}}\big) \big)$ with some finite constants $C_1$ and $C_2$. Now it is enough to show that $h_3(u)$ is integrable. Therefore, consider the quantity
\begin{align*}
\int_0^\infty h_3(u) \ du & = \int_0^\infty C_1 e^{-\frac{u}{\alpha}}  \big( 1 - e^{-\frac{u}{\alpha}} \big)^{\beta-1} \big( C_2 + \ln\big(1 - e^{-\frac{u}{\alpha}}\big) \big) \ du & \\
= & ~ C_1 C_2 \int_0^\infty  e^{-\frac{u}{\alpha}}  \big( 1 - e^{-\frac{u}{\alpha}} \big)^{\beta-1} \ du   +  C_1 \int_0^\infty e^{-\frac{u}{\alpha}}  \big( 1 - e^{-\frac{u}{\alpha}} \big)^{\beta-1} \ln\big(1 - e^{-\frac{u}{\alpha}}\big)  \ du & \\
= & ~ C_1 C_2  \alpha \int_0^1  u_1^{\beta-1} \ du_1   +  C_1 \alpha \int_0^1 u_1^{\beta-1} \ln(u_1) \ du_1 & 
\end{align*} 
Now, when $\beta\leq1$, we see that the integrand of the second term has singularity point at the boundary of the integration domain, but the integrand is integrable because
\begin{align*}
\lim_{\epsilon \to 0^+}	\int_\epsilon^1  u^{\beta-1} \ln(u) \ du & =  \lim_{\epsilon \to 0^+}	 \int_\epsilon^1  \ln(u) \ du  = -1.& 
\end{align*} and hence, after simplification	
\begin{align*}
\int_0^\infty h_3(u) \ du =  & ~ C_1 C_2  \frac{\alpha}{\beta} -  C_1 \frac{\alpha}{\beta^2} < \infty. & 
\end{align*}
Hence, part (\textit{ii}) of Theorem 16.8 of \cite{billingsley2008probability} implies that the likelihood function $\ell_v(\alpha, \beta)$ is differentiable with respect to $\beta$ and its derivative is given by
\begin{align}
\label{eqq7}
\frac{\partial }{\partial \beta} &\ell_v(\alpha, \beta)  =  n! \int_{0}^{\infty}  \frac{\partial}{\partial \beta} e^{h_v(\alpha,\beta;u)} \ du & \nonumber \\
& = n! \Big(\frac{\beta}{\alpha}\Big)^n \int_{0}^{\infty}  \Big( \frac{n}{\beta} + \sum_{i = 1}^{n} \ln\big(1 - e^{-\frac{u+v_i}{\alpha}}\big) \Big)  e^{ -\frac{1}{\alpha}\sum_{i=1}^{n} (u+v_i) } \prod_{i = 1}^{n} \big( 1 - e^{-\frac{u+v_i}{\alpha}} \big)^{\beta-1}. &
\end{align} 
Now, in a similar manner, we step forward to show that $\ell_v(\alpha, \beta)$ is differentiable with respect to $\alpha$ by demonstrating that $\frac{\partial}{\partial \alpha} e^{h_v(\alpha, \beta;u)} $ exists and it is integrable with respect to the variable $u$ provided $0 < v_2 < \dots < v_{n} < \infty, ~ v_{1} = 0 ~ \text{and} ~ \beta>0$.
The first point is obvious to satisfy as it is a function of exponential, logarithmic and polynomial functions which are well-known smooth functions and the derivative is given by 
\begin{align}
\label{eqq8}
\frac{\partial}{\partial \alpha} & e^{h_v(\alpha,\beta;u)}  = e^{h_v(\alpha,\beta;u)} \frac{\partial}{\partial \alpha} h_v(\alpha,\beta;u)  &  \nonumber \\
& =  \Big( - \frac{n}{\alpha} +  \frac{1}{\alpha} \sum_{i = 1}^{n} \frac{u+v_i}{\alpha} - \frac{(\beta-1)}{\alpha} \sum_{i = 1}^{n} \frac{u+v_i}{\alpha}  \Big( \frac{e^{-\frac{u+v_i}{\alpha} }}{1 - e^{-\frac{u+v_i}{\alpha} }} \Big)  \Big) e^{h_v(\alpha,\beta;u) }. &
\end{align} 
Equation \eqref{eqq8} is simplified as follows:
\begin{align}
\label{eqq9}
\frac{\partial}{\partial \alpha} e^{h_v(\alpha,\beta;u)}  
= & - \frac{n}{\alpha} e^{h_v(\alpha,\beta;u) } + \frac{1}{\alpha} \Big(  \sum_{i = 1}^{n} \frac{u+v_i}{\alpha}  \Big)  e^{h_v(\alpha,\beta;u) }  -  \frac{(\beta-1)}{\alpha}  \Big(  \sum_{i = 1}^{n} \frac{u+v_i}{\alpha}  \Big( \frac{e^{-\frac{u+v_i}{\alpha} }}{1 - e^{-\frac{u+v_i}{\alpha} }} \Big) \Big)  e^{h_v(\alpha,\beta;u) } & \nonumber \\
= & - \frac{n}{\alpha} e^{h_v(\alpha,\beta;u) } + \frac{\beta}{\alpha} \Big(  \sum_{i = 1}^{n} \frac{u+v_i}{\alpha}  \Big)  e^{h_v(\alpha,\beta;u) }  -  \frac{(\beta-1)}{\alpha}  \Big(  \sum_{i = 1}^{n} \frac{u+v_i}{\alpha}  \Big(  \frac{1}{1 - e^{-\frac{u+v_i}{\alpha} }} \Big) \Big)  e^{h_v(\alpha,\beta;u) } & \nonumber \\
= & - \frac{n}{\alpha} e^{h_v(\alpha,\beta;u) } + \frac{\beta}{\alpha^2} \Big(  \sum_{i = 1}^{n} v_i  \Big)  e^{h_v(\alpha,\beta;u) } + \frac{n \beta}{\alpha^2} u  e^{h_v(\alpha,\beta;u) } & \nonumber \\
& -  \frac{(\beta-1)}{\alpha^2}  \Big(  \sum_{i = 1}^{n}  \Big(  \frac{1}{1 - e^{-\frac{u+v_i}{\alpha} }} \Big) \Big)  u e^{h_v(\alpha,\beta;u) }  -  \frac{(\beta-1)}{\alpha^2}  \Big(  \sum_{i = 2}^{n} v_i \Big(  \frac{1}{1 - e^{-\frac{u+v_i}{\alpha} }} \Big) \Big)  e^{h_v(\alpha,\beta;u) }. & 
\end{align}
It is clear that first two terms of equation \eqref{eqq9} are integrable as $e^{h_v(\alpha,\beta;u) }$ is integrable with respect to $u$. Now, we show integrablity of the third term and the proofs for the remaining terms goes along the same line and hence proofs are omitted. Assume that 
\begin{align*}
\mathfrak{H}_2(u) = u e^{h_v(\alpha,\beta;u) } = \Big(\frac{\beta}{\alpha}\Big)^n  u \prod_{i = 1}^{n} \Big\{ e^{-\frac{u+v_i}{\alpha}}  \big( 1 - e^{-\frac{u+v_i}{\alpha}} \big)^{\beta-1} \Big\}.
\end{align*} 
For every $\alpha > 0, ~\beta > 0, ~ 0 < v_2 < \dots < v_{n} < \infty, ~ v_1 = 0 $ and $0 < u < \infty $, $\mathfrak{H}_2(u)$ tends to zero as $u \to \infty$. Also, $\mathfrak{H}_2(u) < h_4(u)$, where $ h_4(u) = C_3 u e^{-\frac{u}{\alpha}}  \big( 1 - e^{-\frac{u}{\alpha}} \big)^{\beta-1} $ with some finite constant $C_3$. Now it is enough to show that $h_4(u)$ is integrable. Therefore, we consider the quantity
\begin{align*}
\int_0^\infty h_4(u) \ du & = \int_0^\infty C_3 u e^{-\frac{u}{\alpha}}  \big( 1 - e^{-\frac{u}{\alpha}} \big)^{\beta-1} \ du & \\
= & - C_3  \alpha^2 \int_0^1  u_1^{\beta-1} \ln(1-u_1) \ du_1 &\\
= & - C_3  \alpha^2 \int_0^1  (1-y)^{\beta-1} \ln(y) \ dy &\\
= &  - C_3  \alpha^2 \Big( \int_0^{1/2}  (1-y)^{\beta-1} \ln(y) \ dy + \int_{1/2}^1  (1-y)^{\beta-1} \ln(y) \ dy \Big). &. 
\end{align*}
We denote the first and the second terms of the above equation by $I_1$ and $I_2$, respectively. Here
\begin{align*}
I_1  = \int_0^{1/2}  (1-y)^{\beta-1} \ln(y) \ dy <  \int_0^{1/2} \ln(y) \ dy = \frac{1}{2} ( \ln(1/2) - 1) < \infty
\end{align*} and
\begin{align*}
I_2 & = \int_{1/2}^1  (1-y)^{\beta-1} \ln(y) \ dy =  \frac{1}{2^\beta} \ln(1/2) + \int_{1/2}^1 \frac{ (1-y)^{\beta} }{y} \ dy  & \\
& < \frac{1}{2^\beta} \ln(1/2) + 2 \int_{1/2}^1 (1-y)^{\beta} \ dy = \frac{1}{2^\beta} (\ln(1/2) - \frac{1}{\beta}) < \infty. & 
\end{align*}
Therefore, $\int_0^\infty h_4(u) \ du < \infty$. Hence, it implies that the likelihood function $\ell_v(\alpha, \beta)$ is differentiable with respect to $\alpha$ and its derivative is given by
\begin{align}
\label{eqq10}
\frac{\partial }{\partial \alpha} \ell_v(\alpha, \beta)  & =  n! \int_{0}^{\infty}  \frac{\partial}{\partial \alpha} e^{h_v(\alpha,\beta;u)} \ du & \nonumber \\
& = n! \Big(\frac{\beta}{\alpha}\Big)^n \int_{0}^{\infty}  \Big( - \frac{n}{\alpha} +  \frac{1}{\alpha} \sum\limits_{i = 1}^{n} \frac{u+v_i}{\alpha} - \frac{(\beta-1)}{\alpha} \sum\limits_{i = 1}^{n} \frac{u+v_i}{\alpha}  \Big( \frac{e^{-\frac{u+v_i}{\alpha} }}{1 - e^{-\frac{u+v_i}{\alpha} }} \Big) \Big) &  \nonumber \\
& \qquad \times  e^{ -\frac{1}{\alpha}\sum_{i=1}^{n} (u+v_i) } \prod_{i = 1}^{n} \big( 1 - e^{-\frac{u+v_i}{\alpha}} \big)^{\beta-1}.& 
\end{align} 
\section{Tables}

\begin{table}[h]
\caption{Biases and RMSEs of the estimators while varying sample size $n$ based on $10,000$ simulations.}
\scriptsize
\begin{center}
	\begin{adjustbox}{width=0.79\textwidth, height = 0.59\textwidth}
\begin{tabular}{ccccrrrrrrr}
			\toprule
			\multirow{2}{*}{$\beta$} & \multirow{2}{*}{{$\beta_U$}\tablefootnote[4]{$\beta_U$ is an upper bound for the estimates of shape parameter such that the estimates greater than $\beta_U$ are rejected and the rejected proportion is $p$ reported in the last column of the tables.}}& \multirow{2}{*}{$n$} & \multirow{2}{*}{Method} & \multicolumn{2}{c}{Shape} &  \multicolumn{2}{c}{Scale} & \multicolumn{2}{c}{Location} & \multirow{2}{*}{$p$}  \\ \cmidrule{5-10}
			&                    &  & &   Bias & RMSE                  &                    Bias & RMSE    & Bias & RMSE &  \\ \midrule
			\multirow{32}{*}{0.50}       & \multirow{32}{*}{2}       &          \multirow{4}{*}{20}      
			& LPF   &  0.1706 & 0.2750 & -0.1296 &   0.3473 &  -0.0089  &  0.0171  & 0.0047	\\
			& & & LSPF   & 0.1468 & 	0.2450 & 	-0.1220 & 	0.3310 & 	-0.0078	 & 0.0154 & 0.0036 \\
			& &   & MMLE I   & 0.2642 & 0.3742 & -0.1640 & 0.3538 & -0.0080 & 0.0146 & 0.0076 \\  
			& &   & MMLE III   & 0.0709   &  0.2128  &  -0.0368  &   0.3601   &  0.0005   &  0.0103   &  0.0039 \\ [1.5ex]
			& &  \multirow{4}{*}{50}   
			& LPF  & 0.0885  &  0.1335 &  -0.0879   & 0.2282 &  -0.0017  &  0.0026 & 0.0000 \\
			& & &  LSPF   &  0.0885 & 	0.1265 & 	-0.0968	 & 0.2204 & 	-0.0016 & 	0.0025  & 0.0000 \\
\end{tabular}
	\end{adjustbox}
\end{center}
\label{table1}
\end{table}

\begin{table}
\caption{Biases and RMSEs of the estimators while varying sample size $n$ based on $10,000$ simulations.}
\scriptsize
\begin{center}
	\begin{adjustbox}{width=0.81\textwidth, height = 0.66\textwidth}
		\begin{tabular}{ccccrrrrrrr}
			\toprule
			\multirow{2}{*}{$\beta$} & \multirow{2}{*}{$\beta_U$}& \multirow{2}{*}{$n$} & \multirow{2}{*}{Method} & \multicolumn{2}{c}{Shape} &  \multicolumn{2}{c}{Scale} & \multicolumn{2}{c}{Location}  & \multirow{2}{*}{$p$} \\
			\cmidrule{5-10}
			&                    & & &    Bias & RMSE                  &                    Bias & RMSE    & Bias & RMSE  &          \\ \midrule
			\multirow{40}{*}{1.5}       &   \multirow{40}{*}{12}       &      \multirow{5}{*}{20}        
			& LPF   & 0.5140  &  1.6165  &  0.0070  &  0.2802  & -0.0533   & 0.2121 & 0.0752 \\
			&&&  LSPF   & 0.4277 & 	1.5597 & 	0.0111 & 	0.2937	 & -0.014 & 	0.1566 &0.0756 \\
			& &   & MLE   & -0.4644 & 	0.7732 & 	0.2503 & 	0.4257 & 	0.1141 & 	0.1571 &0.0700	\\
			& &   & MMLE I   &   0.8335 & 1.6122 & -0.0944 & 0.2706 & 0.0084 & 0.1295 & 0.0041	\\
			& &   & MMLE III   & 0.3272  &  1.4925  &  0.0396  &  0.3148  &  0.0060   & 0.1893   & 0.1010 \\ [1.5ex]
			& &   \multirow{5}{*}{50}       
			& LPF  &0.2435  &  0.8945 &  -0.0101  &  0.1804  & -0.0245  &  0.1024 & 0.0040
			\\
\end{tabular}
	\end{adjustbox}
\end{center}
\label{table2}
\end{table}	
\begin{table}
\caption{Biases of the quantile estimators while varying sample size $n$ based on $10,000$ simulations.}
\scriptsize
\begin{center}
	\begin{adjustbox}{width=0.85\textwidth, height = 0.66\textwidth}
		\begin{tabular}{cccrrrrrrrrr}
			\toprule
			\multirow{2}{*}{$\beta$} &  \multirow{2}{*}{$n$} & \multirow{2}{*}{Method} & \multicolumn{8}{c}{$\zeta$} &    \\ \cmidrule{4-12}
			&                    &   &  0.01 & 0.05 & 0.10 & 0.25 & 0.50 & 0.75 & 0.90 & 0.95 & 0.99  \\ \midrule
			\multirow{32}{*}{0.50}        &        \multirow{5}{*}{20}       
			& LPF  & -0.0070  & 0.0005  & 0.0098 &  0.0332  & 0.0501 &  0.0195 &  -0.0675  & -0.1473  & -0.3479\\
			& & LSPF & -0.0063 &	0.0001 &	0.0082 &	0.0292 &	0.0442 &	0.0153 &	-0.0666 &	-0.1417 &	-0.3306 \\
			&  & MMLE I   &  -0.0044 & 0.0081 & 0.0229 & 0.0602 & 0.0921 & 0.0627 & -0.0424 & -0.1418 & -0.3946  	\\ 
			&   & MMLE III  &  0.0015 & 0.0052 & 0.0094 & 0.0169 & 0.0147 & -0.0060 & -0.0387 & -0.0641 & -0.1233 \\ [1.5ex]
			&  \multirow{5}{*}{50}   
			& LPF  & -0.0012  & 0.0021  & 0.0071  & 0.0218  & 0.0348 &  0.0170 &  -0.0401 &  -0.0935 &  -0.2292 \\
			& & LSPF & -0.0012 &	0.0020 &	0.0070 &	0.0218 &	0.0345 &	0.0138 &	-0.0495 &	-0.1086 &	-0.2581 \\
\end{tabular}
	\end{adjustbox}
\end{center}
\label{table3}
\end{table}
\begin{table}
\caption{Biases of the quantile estimators while varying sample size $n$ based on $10,000$ simulations.}
\scriptsize
\begin{center}
	\begin{adjustbox}{width=0.85\textwidth, height = 0.66\textwidth}
		\begin{tabular}{cccrrrrrrrrr}
			\toprule
			\multirow{2}{*}{$\beta$} &  \multirow{2}{*}{$n$} & \multirow{2}{*}{Method} & \multicolumn{8}{c}{$\zeta$} &    \\ \cmidrule{4-12}
			&                    &   &  0.01 & 0.05 & 0.10 & 0.25 & 0.50 & 0.75 & 0.90 & 0.95 & 0.99  \\ \midrule
			\multirow{40}{*}{1.5}        &        \multirow{5}{*}{20}      
			&   LPF   & -0.0030  &	0.0068  &	0.0077  &	0.0033 &	 -0.0050 &	 -0.0098  &	-0.0086 &	 -0.0054  &	0.0046 \\
			& & LSPF & 0.0235 &	0.0251 &	0.0208 &	0.0077 &	-0.0082	 &-0.0165 &	-0.0147	 &-0.0096 &	0.0063 \\
			&  & MLE & 	0.0873 &		0.0333 &		-0.0083 &		-0.0783 &		-0.1076	 &	-0.0311 &		0.1494 &		0.3078 &		0.6989
			\\ 
			&  & MMLE I   &  0.0788 & 0.1057 & 0.1156 & 0.1189 & 0.0961 & 0.0395 & -0.0431 & -0.1074 & -0.2584	\\
			&   &MMLE III  & 0.0399 & 0.0354 & 0.0263 & 0.0045 & -0.0164 & -0.0177 & 0.0037 & 0.0265 & 0.0866 \\ [1.5ex]
			&  \multirow{5}{*}{50}   
			&   LPF    & -0.0007  &	0.0060  &	0.0076 &	 0.0065 &	 0.0009  &	-0.0080 &	 -0.0184  &	-0.0258 &	 -0.0423 \\
\end{tabular}
	\end{adjustbox}
\end{center}
\label{table4}
\end{table}
\begin{table}
\caption{RMSEs of the quantile estimators while varying sample size $n$ based on $10,000$ simulations.}
\scriptsize
\begin{center}
	\begin{adjustbox}{width=0.85\textwidth, height = 0.66\textwidth}
		\begin{tabular}{cccrrrrrrrrr}
			\toprule
			\multirow{2}{*}{$\beta$} &  \multirow{2}{*}{$n$} & \multirow{2}{*}{Method} & \multicolumn{8}{c}{$\zeta$} &    \\ \cmidrule{4-12}
			&                    &   &  0.01 & 0.05 & 0.10 & 0.25 & 0.50 & 0.75 & 0.90 & 0.95 & 0.99  \\ \midrule
			\multirow{32}{*}{0.50}        &        \multirow{4}{*}{20}       
			&    LPF   & 0.0134  &	0.0096  &	0.0207  &	 0.0609  &	 0.1284  &	0.2614  &	  0.5069   &	0.7214   &	 1.2550 \\
			& & LSPF & 0.0126 &	0.0096 &	0.0192 &	0.0566 &	0.1243 &	0.2596 &	0.4993 &	0.7054 &	1.2152 \\
			&  & MMLE I   &   0.0101 & 0.0171 & 0.0362 & 0.0889 & 0.1635 & 0.2822 & 0.5117 & 0.7217 & 1.2559	\\
			&   & MMLE III  & 0.0089 & 0.0121 & 0.0217 & 0.0543 & 0.1185 & 0.2540 & 0.5049 & 0.7255 & 1.2771 \\ [1.5ex]
			&  \multirow{4}{*}{50}   
			&    LPF   & 0.0020  &	0.0045 &	 0.0122 &	 0.0367  &	0.0794 &	 0.1638  &	 0.3232  &	 0.4635  &	 0.8136
			\\
			& & LSPF & 0.0020 &	0.0042 &	0.0115 &	0.0356 &	0.0789 &	0.1638 &	0.3194 &	0.4552 &	0.7936 \\
\end{tabular}
	\end{adjustbox}
\end{center}
\label{table5}
\end{table}
\begin{table}
\caption{RMSEs of the quantile estimators while varying sample size $n$ based on $10,000$ simulations.}
\scriptsize
\begin{center}
	\begin{adjustbox}{width=0.85\textwidth, height = 0.66\textwidth}
		\begin{tabular}{cccrrrrrrrrr}
			\toprule
			\multirow{2}{*}{$\beta$} &  \multirow{2}{*}{$n$} & \multirow{2}{*}{Method} & \multicolumn{8}{c}{$\zeta$} &    \\ \cmidrule{4-12}
			&                    &   &  0.01 & 0.05 & 0.10 & 0.25 & 0.50 & 0.75 & 0.90 & 0.95 & 0.99  \\ \midrule
			\multirow{40}{*}{1.5}        &        \multirow{5}{*}{20}       
			&   LPF  &  0.1122  &	0.0944 &	 0.1000 &	 0.1327  &	 0.1979  &	 0.3294 &	  0.5437  &	 0.7210  &	1.1528 \\
			& & LSPF & 0.1055 &	0.1052 &	0.1146 &	0.1438 &	0.2004	 &0.3301 &	0.5523 &	0.7378 &	1.1904 \\
			&  & MLE & 	0.1284	& 	0.1014	& 	0.1050	& 	0.1531	& 	0.2118	& 	0.3227		& 0.5982	& 	0.8551	& 	1.5038	 \\ 
			&  & MMLE I   &  0.1260 & 0.1527 & 0.1692 & 0.195 & 0.2337 & 0.3402 & 0.5386 & 0.7073 & 1.1213	\\
			&   &MMLE III  & 0.1132 & 0.0999 & 0.1070 & 0.1432 & 0.2060 & 0.3293 & 0.5518 & 0.7447 & 1.2240 \\ [1.5ex]
			&  \multirow{5}{*}{50}  
			&    LPF   & 0.0508  &	0.0521 &	 0.0637  &	0.0893  &	0.1269  &	 0.2075  &	 0.3447  &	 0.4591  &	 0.7374 \\
\end{tabular}
	\end{adjustbox}
\end{center}
\label{table6}
\end{table}
\end{document}